\documentclass[sigconf, nonacm]{acmart}

\newcommand\vldbdoi{XX.XX/XXX.XX}

\newcommand\vldbavailabilityurl{URL_TO_YOUR_ARTIFACTS}

\usepackage{amsmath}
\usepackage{algorithm}
\usepackage[noend]{algpseudocode} 
\usepackage{tabularx}
\usepackage{subcaption}

\definecolor{nischalcolor}{RGB}{220,68,5}
\newcommand{\nischal}[1]{\textcolor{nischalcolor}{#1}}
\newcommand{\ara}[1]{\textcolor{red}{#1}}

\makeatletter
\renewcommand\paragraph{%
  \@startsection{paragraph}{4}{\z@}%
    {1.25ex \@plus1ex \@minus.1ex}
    {0em}
    {\normalfont\normalsize\bfseries}
}
\makeatother

\newcommand{\squishlisttwo}{
   \begin{list}{$\bullet$}
       { \setlength{\itemsep}{0pt}      \setlength{\parsep}{3pt}
       \setlength{\topsep}{3pt}       \setlength{\partopsep}{0pt}
      \setlength{\leftmargin}{1.5em} \setlength{\labelwidth}{1em}
       \setlength{\labelsep}{0.5em} } }

\newcommand{\squishlistthree}{
   \begin{list}{$\bullet$}
   { \setlength{\itemsep}{0pt}    \setlength{\parsep}{0pt}
   \setlength{\topsep}{0pt}     \setlength{\partopsep}{0pt}
   \setlength{\leftmargin}{2em} \setlength{\labelwidth}{1.5em}
   \setlength{\labelsep}{0.5em} } }
\newcommand{\squishend}{
\end{list}  }

\usepackage{amsthm}
\usepackage{mdframed}

\newtheoremstyle{problemstyle}
  {1em}{1em}{\itshape}{}{\bfseries}{.}{.5em}{}

\theoremstyle{problemstyle}
\newtheorem{innerproblem}{Problem}

\newenvironment{problem}
  {\vspace{0.8em}
   \begin{mdframed}[linewidth=0.8pt,roundcorner=4pt,
                    backgroundcolor=white,
                    innerleftmargin=5pt,innerrightmargin=5pt,
                    innertopmargin=3pt,innerbottommargin=3pt]
   \begin{innerproblem}}
  {\end{innerproblem}\end{mdframed}\vspace{0.8em}} 

\begin{document}
\title{Querying with Conflicts of Interest}

\author{Nischal Aryal}
\affiliation{%
  \institution{Oregon State University}
}
\email{aryaln@oregonstate.edu}

\author{Arash Termehchy}
\affiliation{%
  \institution{Oregon State University}
}
\email{termehca@oregonstate.edu}

\author{Marianne Winslett}
\affiliation{%
  \institution{University of Illinois}
}
\email{winslett@illinois.edu}

\begin{abstract}
Conflicts of interest often arise between data sources and their users regarding how the users' information needs should be interpreted by the data source. For example, an online product search might be biased towards presenting certain products higher than in its list of results to improve its revenue, which may not follow the user's desired ranking expressed in their query. 
The research community has proposed schemes for data systems to implement to ensure unbiased results. However, data systems and services usually have little or no incentive to implement these measures, e.g., these biases often increase their profits.
In this paper, we propose a novel formal framework for querying in settings where the data source has incentives to return biased answers intentionally due to the conflict of interest between the user and the data source. 
We propose efficient algorithms to detect whether it is possible for users to extract relevant information from biased data sources. We propose methods to detect biased information in the results of a query efficiently.
We also propose algorithms to reformulate input queries to increase the amount of relevant information in the returned results over biased data sources.
Using experiments on real-world datasets, we show that our algorithms are efficient and return relevant information over large data.
\end{abstract}

\maketitle



\section{Introduction}

\label{section:introduction}
Due to financial, social, and political factors, the incentives of data source owners (data sources for short) and their users are often not perfectly aligned.
%
%
For example, data sources that provide information about products may benefit financially from convincing users to buy certain products. In fact, Amazon, Google, and Apple have been accused of returning results that are biased in this manner \cite{Mattioli2019, Fussell2019, Mickle2019, Nicas2019}.
As an example, Apple allegedly reworked its ranking algorithm so that Apple Music was ranked first and Spotify 23rd in the App Store when a user searched for music apps \cite{Nicas2019}. Similarly, a search for podcast apps returned Apple Podcasts in the first place, followed by 12 irrelevant Apple apps, and then finally podcast apps from other publishers. A search for ``TV'' in 2018-9 would return Apple apps at the top, and Netflix ranked in the 100s \cite{Nicas2019}. As experts commented at the time: “I find it hard to believe that organically there are certain Apple apps that rank better than higher-reviewed, more downloaded competitors,'' and ``there’s just a ton to be gained commercially by dominating'' search results \cite{Nicas2019}. 

This conflict of interest is inherent at some data sources, such as shopping websites, which charge sellers a `referral fee' of considerable fraction of the price of items sold on its site~\cite{AmazonReferralFees}, so that it profits more when users buy higher-priced versions of a given product. Cost-conscious users who re-sort the results by price may have to wade through many undesirable results before reaching the first good match.


\begin{example}
\label{example:headsets}
On the day of this writing, a search for `headset' on a well-known shopping website ({\bf SW} for short) returns headphones priced \$24.99 to \$39.95 in its top `Results' section,  and all items on the remainder of the page are indeed headphones. When re-sorted by price, the top hit on the first page is a charging cord, followed by earbuds, microphones, cables, and one \$1.98 headset. Only on the seventh page of results do headphones start to outnumber other accessories. 
\end{example}

Conflicts of interest may also arise from non-financial considerations. 
To maximize user engagement, some social media sites may promote extreme content, which is not necessarily in anyone else's best interest.
More generally, data sources may return information that is biased toward or against controversial topics and causes, political groups, or demographic groups, e.g., not returning or giving a low rank to information related to these \cite{10.1145/3533379,10.1145/3533380}. Perhaps the best known recent example is the removal of all topics related to climate change from US government websites \cite{Zraik2025}.
Data sources that belong to foreign adversaries can also promote content that serves the objectives of the data source owner.

In the aforementioned settings, data sources interpret the information in user queries to take actions that do not exactly meet users' preferences and objectives and, in fact, might largely deviate from what users deem as desirable.
Current proposals to combat undesirable data source actions usually require the data source to follow and implement certain protocols \cite{10.1145/3488717}. 
For example, researchers have proposed methods to remove bias and ensure fairness for data sources to implement \cite{10.1145/3533379,10.1145/3533380,10.1145/3132847.3132938,10.1145/3366424.3380048,tabibian20a,mathioudakis2020affirmative}.
However, data sources often do not have any incentive to implement these proposals because these are usually against their objectives and business models. 
Some might also argue that these restrictions limit the freedom to disseminate information or conduct business and may not be ethical or have long-term adverse impacts on our society.

In such a situation, the data source and users do share some interests: the data source would like to satisfy the information needs behind users' queries well enough to keep them engaged. 
Users can leverage this partial common interest by modifying their queries so that the data source interpretations of the submitted queries are those desired by users.

\begin{example}
\label{example:bias-query}
Consider again the user in Example~\ref{example:headsets} 
who looks for affordable headphones on SW. 
If the results are biased towards relatively more expensive products, she can add a price cap to her query.
For example, when we reissued the SW headset query with a \$19 maximum price, the {\it second} result returned cost \$21, but the others on the first page cost less than \$19. 
Thus, SW does not take price caps literally, but they do influence the results. 
\end{example}

But two can play at this game. 
The data source may be aware of the fact that users know about its undesirable actions, e.g., users have observed undesirable biases in their previous experiences of using the data source.
Thus, it may reason that users aim at guiding its interpretations and results towards the users' desired results by modifying their original queries.
Since it is not in the data source's best interest to accept the modified query as precisely the user's true information need, the data source may try to recover the original intent behind the modified query.
For instance, it might also use its prior belief about users' intents, e.g., gained from external sources or previous interactions with users, to interpret the true intent behind the modified query.
Consequently, users may reason about the data source reasoning process and try to modify their queries more drastically to influence the data source interpretation. 
The user and data source reason and make decisions based on possible reasoning and strategies of each other iteratively to find the best method to formulate and interpret queries.

\begin{example}
\label{example:bias-reason}
In Example~\ref{example:bias-query}, given that the user has added a degree of bias in the modified query, the data source may reason that the user has overstated their preference for headphones of a particular price, just to offset the data source's bias.
Hence, it might avoid placing such headphones in the top positions, as we saw when SW still ranked a \$21 headset second even after we limited the price to \$19 in Example \ref{example:headsets}.
%
The user might counteract this by further reducing her submitted maximum price for a headset.
\end{example}

In some of the aforementioned settings, a brute-force approach for the user to get her desired information without allowing the data source to return undesirable results is to submit a query that retrieves most or all of the data and post-process that to extract her desired information. 
In extreme, in Example~\ref{example:bias-reason}, the user who searches the data source might submit a query that returns all products. 
This is not usually practical, because of the enormous amount of data managed by these data sources.
Furthermore, data sources might not return such a large dataset as it is against their objective, e.g., they might ask users to resubmit a query that returns a relatively small subset of the data.
In addition, data sources usually limit the rate of retrieved information, making it difficult for users to download a very large dataset using multiple queries.
Moreover, users often lack the computational resources, time, and skills to store and post-process a large dataset. 

There are four important problems in this monopolistic setting. First, we would like to know under what conditions the aforementioned recursive reasoning of the user and the data source will converge to a stable state in which the user finds a query that gives her useful information. 
The user might be able to leverage the fact that the data sources would like to keep the user engaged, i.e., partial common interest, to convince the data source to return her desired results to some degree. 
Second, we would like to inform users about the untrustworthy information in the results returned by the data source, e.g., whether the returned top-k results are in fact the top-$k$ results for the query on the entire data. 
It is challenging as the data source may not show all relevant information to users so they can check the validity of the information.
Third, given a user query, we would like to know how much reliable information the user can get from the data source.
Finally, given a query, we would like to find another query that can convince the data source to return more relevant and reliable information than the original query.

In this paper, we propose a framework that models strategic information sharing between a user and a data source when there are conflicts of interest between them.
The framework models the user and the data source as agents with different objectives and uses the current body of research in game theory on how agents with different goals establish informative communication \cite{10.2307/1913390,CHAKRABORTY200770}.
Using this framework, we propose efficient algorithms that solve the aforementioned problems efficiently.
Our contributions are as follows:
\squishlisttwo
    \item We propose a formal framework that models the communication of the user and the data source with conflicts of interest using queries that return ranked lists of results (Section~\ref{sec:framework}). Our framework defines the stable states (fixed points) in the recursive reasoning of the users, e.g., client or middleware systems, and data source with conflict of interest for answering a query.
    
    \item If the conflict of interest between the user and the data source is too high, the data source may ignore the user query and will return its own preferred ranking regardless of the preference expressed in the input query. We propose efficient algorithms that check whether the communication has stable states in which the input query can influence the decision of the data source about returning the answers (Section~\ref{sec:possibility-of-informative-interaction}).
    
    \item We propose an efficient algorithm that detects untrustworthy results in the answer to a query  (Section~\ref{sec:possible-informative-equilibria}).
    
    \item We prove that the problem of finding the query that returns the maximum amount of information for an original user intent is NP-hard.
    Hence, we propose efficient algorithms to solve this problem for some general classes of utility functions (Section~\ref{sec:maximum}).
    
    \item We conduct extensive empirical studies using a wide variety of real-world datasets. Our empirical studies indicate that our algorithms scale to large data (Section~\ref{sec:experiments}). 
\end{list}

\section{Framework}
\label{sec:framework}

\subsection{Data Model \& Query Language}
\label{sec:data-model-query-language}

A schema $\mathcal{R}$ is a set of relation symbols $R_i$, $1\leq i \leq n$, where each $R_i$ has a fixed arity $a_i$.  
Each relation symbol $R_i$ is associated with a set of attribute symbols $(A_{i,j}$, $1 \leq j \leq a_i$).
To simplify our exposition, we define the domain of possible values for each attribute $\mathbf{dom}\bigl(A_{i,j}\bigr)$ in each relation as a subset of or equal to a single set $\mathbf{dom}$.
Each tuple of the relation symbol $R_i$ is then an element of the Cartesian product
$
  \prod_{j=1}^{z_i}\mathbf{dom}\bigl(A_{i,j}\bigr)
$. A relation of $R_i$ is a finite set of such tuples. 
An instance of $\mathcal{R}$ is a set of relations for its relation symbols.



Let $\mathcal{Q}$ denote a query language.  
An $a$-ary query $q \in \mathcal{Q}$, $a > 0$, in schema $\mathcal{R}$ is a function that maps each instance $I$ of $\mathcal{R}$ to a totally ordered list where each member of the list belongs to $\bf{dom}^{z}$. 
We denote the result of the query $q$ over $I$ by $q(I)$.
Let $elm(l)$ denote the set of tuples in a totally ordered list $l$. For $e,e' \in elm(q(I))$, we use $e \prec e'$ ($e \sim e'$) to show that $e$ precedes $e'$ ($e$ and $e'$ share the same position) in $q(I)$. 
For brevity, by abuse of notation, we show $e \in elm(q(I))$ as $e \in q(I)$.
A pairwise order $e \prec e'$ is in $q(I)$ if the order is present in the result of $q(I)$. For brevity, we show $e \prec e'$ is in $q(I)$ as $e \prec e' \in q(I)$.

Users can express queries in our framework using query languages, such as SQL, that return a set or a ranked list of results, e.g., using ranking constructs such as the {\it Order By Case} clause~\cite{casemysql, casesqlserver, oracle,illyas2008topksurvey,illyas2008topksurvey,10.5555/645926.671875}.
They can also use proposed constructs to express preference relationships in query languages \cite{KIELING2002990,10.5555/876881.879661,10.1145/2723372.2723736,illyas2008topksurvey}. 



\begin{table}[h]
  \vspace{-3mm}
  \centering
  \caption{\texttt{Headphones}}
  \vspace{-3mm}
  \label{tab:headphones}
  \small
  \setlength{\tabcolsep}{3pt}
  \renewcommand{\arraystretch}{0.95}
  \begin{tabular}{r l l r r c}
    \toprule
    \textbf{id} & \textbf{model} & \textbf{brand} & \textbf{rating} & \textbf{price} & \textbf{best seller} \\
    \midrule
    $e_1$ & Tune 510BT & JBL & 4.5 & 25.20 & Y \\
    $e_2$ & Q20i & Soundcore & 4.6 & 39.99 & N \\
    $e_3$ & A10 Pro & Sisism & 4.8 & 21.99 & N \\
    $e_4$ & Sports & Haoyuyan & 3.9 & 299.99 & N \\
    \bottomrule
  \end{tabular}
  \vspace{-3mm}
\end{table}

\subsection{Intents, Queries, and Results}
\label{sec:actions}
\paragraph{User's Intents and Queries. }
We call the user's true information need, which is expressed in $\mathcal{Q}$, their {\it intent}. The data source does {\it not} know the user's intent. 
The user picks their intents from a prior distribution $\pi$, which is known to the data source. 
Given intent $\tau$, the user submits the query $q \in \mathcal{Q}$ to the data source.
The query $q$ may {\it not} be equivalent to $\tau$, e.g., to hide the user's intent $\tau$ to offset the data source bias in its returned results for $q$.

\begin{example}
\label{example:user-intent}
Consider a user who queries SW to search for a headphone. A user's intent $\tau$ might be to view headphones ordered by customer rating: \texttt{SELECT DISTINCT brand, model, rating, price
FROM Products
WHERE type = `headphones'
ORDER BY rating DESC.} 
To express this in SW's real-world interface, the user would first type `headphones' into the search bar, then ask for them to be sorted by rating ($\textbf{q}_{\textbf{rating}}$). As of this writing, when we submit $\textbf{q}_{\textbf{rating}}$ to SW, the first page of interpretation $\beta$ results includes 16 models representing a variety of brands and price points, as partially shown in Table 1. The query panel on the results page says that it is showing only products with ratings of 4 and up, though the Haoyuyan brand is rated 3.9.
\end{example}

\paragraph{Data Source Interpretation. }
The data source interprets the query $q \in \mathcal{Q}$ submitted by the user to obtain its own desired query $\beta$.
If there is a conflict of interest between the user and the data source, the interpreted query $\beta$ may {\it not} be equivalent to $q$, e.g., if the data source is biased towards ranking certain tuples high.


\subsection{Strategies}
\label{sec:strategies}
\paragraph{User strategy. }
The user strategy $P^r(q | \tau)$ is a (stochastic) mapping from $\mathcal{Q}$ to $\mathcal{Q}$.
Given the intent $\tau$, 
$P^r$ returns queries that the user submits to the data source.


\begin{example}Continuing with Example~\ref{example:user-intent}, to offset SW's bias, we submitted a different query
$\textbf{q'}_{\textbf{rating}}$ that asked for non-major brands (non-JBL)  sorted by rating. This resulted in 4.8 and 4.9-rated products from Skullcandy being listed 12th and 14th on the first page,
even though all the other products on that page were only rated 4.7. Similarly, the first item on the second page was a 4.8 Skullcandy headset that was not listed on the previous page. However, when we searched for `headphones' with SW's default search order, no Skullcandy products appeared in the first 10 pages of results. This inspires us to imagine a situation where a data source is biased against products from a brand $B$ and systematically demotes them in its interpretation of the \texttt{ORDER BY clause}, e.g., \texttt{ORDER BY CASE WHEN brand = `B' THEN rating - .2 ELSE rating END DESC}. In such a situation, if the user knows about this bias and is quite fond of products from $B$, a better submitted query for her intent is $\textbf{q''}_{\textbf{rating}}$ \texttt{ORDER BY CASE WHEN brand = `B' THEN .2 + rating ELSE rating END DESC}. This lists all $B$ products first, in rating order, followed by the other brands' models  in descending order of rating.  In the SW interface, the user could approximate this by first using the search bar to query for ``headphones'', sorting by rating, and then checking the box to view Skullcandy headphones first. When we did this, the first page of results were all Skullcandy products, though they were not actually presented in rating order.
\end{example}


\noindent{\bf Data source strategy. }
The {\it data source strategy} $P^s(\beta | q)$ is a (stochastic) mapping from $\mathcal{Q}$ to $\mathcal{Q}$.
The data source uses $P^s$ to interpret the intent behind the submitted query $q$.
For example, if the data source is biased, given $q$, it computes a biased query $\beta$ and returns the result of $\beta$ on its database instance.

\subsection{Interaction Setting}
\label{sec:utility-functions}
\paragraph{Utility Functions. }
The {\it utility functions} of the user and the data source measure the desirability of the interpretation of the intent of the user for the user and the data source, respectively.
We denote the utility functions of the user and the data source by $U^r(\tau,\beta)$ and $U^s(\tau,\beta)$, respectively, which map the query pairs in $\mathcal{Q} \times \mathcal{Q}$ to real values.
If there is some conflict of interest between the user and the data source, $U^r$ and $U^s$ take their maximum values for different values of $\beta$ for the same intent $\tau$: $U^r$ takes its maximum when $\beta$ and $\tau$ are equivalent, and $U^s$ will be at its maximum for some $\beta'$ that is not equivalent to $\tau$ and returns results that satisfy the biases of the data source. The exact instantiation of $U^r$ and $U^rs$ varies depending on the application and the domain of the data.




\begin{example}
\label{example:utilities-query}
Consider the intent \(\tau \) from Example~\ref{example:user-intent} but \textit{no access to the underlying tuples in the Headphone instance}. Assume only two possible data source interpretations:
$\beta_1:\;$ \texttt{ORDER BY CASE WHEN brand=`Skullcandy' THEN rating - .2 ELSE rating END} and $\beta_2:\;$ \texttt{ORDER BY rating DESC}. Interpretation \(\beta_2\) better matches the user's intent than \(\beta_1\).
 To output a real-valued score, suppose the utility function measures absolute errors in ranking positions (Spearman Footrule distance~\cite{spearman1961l1distance}).
Concretely, assume access to the domain of {\tt rating} based on a known schema \(\mathcal{R}\), where the true ranks under \(\tau\) of the data source's four headphones are \([1,2,3,4]\) and under \(\beta_1\) are \([4,3,2,1]\). Then $
U^r(\tau,\beta_1)
=-\bigl(|1-4|+|2-3|+|3-2|+|4-1|\bigr)
=-8,
$ Similarly, $U^r(\tau,\beta_2)=-2
$, so interpretation \(\beta_2\) provides higher user utility.
\end{example}


\paragraph{Interaction Setting. }
An {\it interaction setting} (interaction for short) is the tuple $(\mathcal{R}, \tau, U^r, U^s)$ where $\mathcal{R}$ is the schema of the database, $\tau$ is the user intent and $U^r$ and $U^s$ are the utility functions of the user and the data source, respectively.
$U^r$ and $U^s$ are the same for all user intents on $\mathcal{R}$. 

\paragraph{Common knowledge. }
We assume that the user knows the schema, but {\it not} the database instance maintained in the data source.
Both data source and use know all utility functions.
Users may learn the utility function of the data source by querying the data source and using their past interactions with the data source.
\subsection{Equilibria of Interaction}
\label{sec:belief-and-euilibrium}

\paragraph{Stable states. }
As explained in Section~\ref{section:introduction}, the user and the data source use their available information to reason about the strategies of the other party to maximize their utility function.
We would like to identify whether this reasoning will converge to any eventual stable state in which the user and data source settle on a pair of fixed strategies where none can increase their utility by deviating from their strategies. 
Knowing these states will enable us to analyze the strategies selected for the user and the data source and design algorithms to find desired strategies for users in these settings. 

\begin{example}
\label{example:credible-communication-stable}
Consider the following scenario. The \texttt{Headphones} instance $I$ has only two tuples
$e_1$ (JBL, rating 4.5, \$25.20) and $e_2$ (Soundcore, rating 4.6, \$39.99).
The user has one of two intents:
price-conscious $\tau(I), e_1 \prec e_2$ or rating-conscious $\tau'(I): e_2 \prec e_1$.
The data source chooses between two interpretations,
$\beta(I): e_2 \prec e_1$ and $\beta'(I): e_1 \prec e_2$.
Suppose the user gets $U^r=2$ if the interpretation matches her intent and $U^r=0$ otherwise. The source earns an extra commission $c>0$ when the expensive headphone $e_2$ is ranked first, and suffers a loss $L>0$
when the user is dissatisfied with the results. The source utility is commission minus dissatisfaction loss.
This yields the payoff matrix:
\[
\begin{array}{c|cc}
      & \beta\ (e_2\prec e_1) & \beta'\ (e_1\prec e_2) \\
\hline
\tau (e_1 \prec e_2\    & (0,\;c-L) & (2,\;0) \\
\tau'\ (e_2 \prec e_1) & (2,\;c)   & (0,\;-L)
\end{array}
\]
where each cell shows $(U^r,U^s)$. If $c<L$, then the source prefers $\beta'$ under $\tau$ (since $0>c-L$) and prefers $\beta$
under $\tau'$ (since $c>-L$).
Hence, there is a stable state. The user strategy is for intent $\tau$ to submit a price query
(e.g., \texttt{ORDER BY price ASC}) and to submit a rating query
(e.g., \texttt{ORDER BY rating DESC}) for $\tau'$. Based on the utility function, the source will follow these queries by choosing
$\beta'$ after the former and $\beta$ after the latter. However, if instead the commission is very large $c\ge L$, then the source weakly prefers $\beta$ even under $\tau$ (since $c-L\ge 0$),
and strictly prefers $\beta$ under $\tau'$. Thus, the source will return $\beta$; this may prompt a new cycle of reasoning.
\end{example}


\paragraph{Bayesian equilibrium. }  
Since the data source does not know the user's intent, we follow the approach used in the game theory literature and use the concept of {\it Bayesian equilibrium} to characterize the stable state of the user and data source reasoning \cite{10.2307/1913390}. The {\it belief of the data source} $\Phi(\tau | q)$ is a stochastic mapping from $\mathcal{Q}$ to $\mathcal{Q}$.
The belief of the data source maps every query submitted by the user $q$ to a probability distribution over the set of possible intents. 
The {\it expected utility of the user} for the query $q$ is 
\begin{equation*}
\mathbb{E}[U^r(\tau,\beta) \mid q] = \sum_{\beta \in \mathcal{Q}} U^r(\tau, \beta) P^s(\beta \mid q).    
\end{equation*}


\noindent
The {\it expected utility of the data source} for the query $q$ is
\begin{equation*}
\mathbb{E}[U^s(\tau, \beta') \mid q] = \sum_{\tau \in \mathcal{Q}} U^s(\tau, \beta') \Phi(\tau \mid q). 
\end{equation*}
where $\Phi(\tau \mid q)$ $=$ $\frac{ P^r(q \mid \tau) \pi(\tau)} {\sum_{\gamma \in Q} P^r(q \mid \gamma) \pi(\gamma)}$.



\begin{definition}
\label{definition:equilibrium}
An equilibrium for interaction $(\mathcal{R}, \tau, U^r, U^s)$ is a pair of user and data source strategies $(P^r(q | \tau), P^s(\beta | q))$ such that 
\begin{itemize}
    \item 
    if $P^r(q \mid \tau) > 0$, $q = \arg\max_{q'} \mathbb{E}[U^r(\tau, \beta) \mid q']$. 
    \item 
    if $P^s(\beta \mid q) > 0$, $\beta = \arg\max_{\beta'} \mathbb{E}[U^s(\tau,\beta') \mid q]$. 
\end{itemize}
\end{definition}




\section{Possibility of Influence}
\label{sec:possibility-of-informative-interaction}
A data source may be so biased towards returning certain order of tuples for a given query such that no user's strategy can modify its decision, i.e., no matter what the input query expressed the user's desired order between tuples, the data source will return tuples based on its own preferred order.
For example, if a shopping website is significantly biased toward returning products with higher prices, no matter what the user's submitted query, e.g., returning products based on their quality ratings, the website will return a fixed order of products based on their prices.
In this section, we investigate these settings and propose efficient algorithms to detect them.

\subsection{Influential Interactions}
\label{sec:uninformative-informative}
An equilibrium $(P^r,P^s)$ is {\bf influential} in the interaction $(\mathcal{R}, \tau, U^r, U^s)$ if and only if there are at least two queries $q,q' \in \mathcal{Q}$, $P^r(q|\tau) > 0$ and $P^r(q'|\tau) > 0$, such that $P^s(\beta|q)$ $\neq$ $P^s(\beta|q')$.
Otherwise, the equilibrium is non-influential.
If an equilibrium is non-influential, the data source selects an interpretation that maximizes its expected utility function based on prior $\pi$ over the user's intent in the equilibrium.
That is, the data source finds it optimal to return answers purely according to its prior belief.


From the definition of equilibrium, it follows that every interaction has a non-influential equilibrium in which the user will send the same query for all their intents.
An interaction is {\bf influential} if it has at least one influential equilibrium and is non-influential otherwise.
For instance, in Example~\ref{example:credible-communication-stable} the stable state is an influential equilibrium, whereas the other is an uninfluential equilibrium.
Given an interaction setting, before extracting reliable information from data source results or finding the desired strategy for the user, we should check whether the interaction is influential.

\begin{problem} 
Given the interaction setting $(\mathcal{R}, \tau, U^r, U^s)$, determine whether the interaction is influential.
\end{problem}

Example~\ref{example:credible-communication-stable} illustrates a setting where the objectives of the user and the data are extremely opposed, therefore, no influential interaction is possible. 
The following result formalizes this intuition.
To simplify our exposition, we assume that intents, submitted queries, and interpretation on schema $\mathcal{R}$ do not return any duplicate tuple over any instance of $\mathcal{R}$. 
Our results extend to the cases where they return duplicate tuples.
Two intents on schema $\mathcal{R}$ are \textbf{set-equivalent} if they return the same set of answers over all instances of $\mathcal{R}$.

\begin{theorem}
\label{theorem:noncred-babbling}
The interaction $(\mathcal{R}, \tau, U^r,U^s)$ is influential if and only if there are set-equivalent intents $\tau \neq \tau'$ and interpretations $\beta \neq \beta'$ where the following conditions hold for both user and data source utility functions $U^t \in \{U^r, U^s\}$: 
$$ U^t(\tau,\beta) \ge U^t(\tau,\beta') \quad \text{and} \quad U^t(\tau',\beta') \ge U^t(\tau',\beta). $$
or
$$ U^t(\tau,\beta') \ge U^t(\tau,\beta) \quad \text{and} \quad U^t(\tau',\beta) \ge U^t(\tau',\beta'). $$
\end{theorem}



The proofs of our theoretical results are in Section~\ref{sec:proofs}.
Theorem~\ref{theorem:noncred-babbling} provides necessary and sufficient conditions on user and data source's utility functions that enable influential interaction.  
We can use this result to check whether the interaction is influential by checking the intents and interpretations in 
$\mathcal{Q}$.
However, it will be too inefficient for expressive intent and query languages.

\subsection{Detecting Influential Interactions Efficiently}
\label{sec:efficiently-detecting-uninformative-setting}
In this section, we investigate settings for which we can efficiently detect whether an influential equilibrium exists. We first place some mild and natural conditions on utility functions. 
Given these conditions, we provide efficient algorithms to detect whether influential equilibria are possible.

\paragraph{Bias Function.} The utility of the data source may vary between tuples. Consider a shopping website that sells products from various brands, including the website's owner. The data source may prefer tuples that contain information about its own products due to higher profitability. 
We formalize this heterogeneity using {\it bias functions}.
Given relation schema $\mathcal{R}$, the bias function $b(e)$ maps each tuple $e$ into the domain of $\mathcal{R}$ to 
real values ($\mathbb{R}$).
If the data source is not biased to tuple $e$, we have $b(e)$ $=$ $0$.
Since the data source is biased towards (or against) some tuples, $b(e) \neq 0$ for some tuples $e$. 
The exact instantiation of the bias function varies according to the application. 
For instance, one may define a bias function as a mapping of one or multiple attributes of tuples, e.g., brand.

\paragraph{Additive Utility.}  
Given an intent $\tau$ and a data source interpretation $\beta$ over an instance $I$ of schema $\mathcal{R}$, we assume that  the utility function for the data source is defined as:
\begin{equation}
\label{equation:utility_ranks}
U^s(\tau,\beta) \;= 
\sum_{e \in \tau(I)} 
u^s\big(\text{r}(e, \tau(I)),\, \text{r}(e, \beta(I)), b(e)\big)
\end{equation}
Here, $\text{r}(e, q(I))$ is a function that returns the position of a tuple $e \in q(I)$ and $u^s(\cdot, \cdot, \cdot)$ is a function that maps pairs of positions and a bias value to a real value. 
We make the same additivity assumption for the user's utility where the bias function maps every tuple to zero, as the user is not biased to any particular tuple.
Due to the biases of the data source against or toward some tuples, there might be tuples in $\beta(I)$ that do not belong to $\tau(I)$.
If the data source does not return some tuples in $\tau(I)$, we set $\text{r}(e, \beta(I))$ to a fixed value greater than the number of tuples in $\tau(I)$.
For the rest of the paper, we assume that all utility functions are in the form of Equation~\ref{equation:utility_ranks}.
Since users usually do not know the exact instance $I$, they cannot accurately compute the value of $U^s$ for $I$. 
However, as we will show, they can reason about the utilities given the schema of $I$.

\begin{example}
\label{ex:quadratic-bias}
Consider the interaction from Example~\ref{example:utilities-query} with concrete utility functions.
The user gets maximum utility when the result perfectly matches their intent. We model this by function:
\[
u^r\big(\mathrm r(e,\tau(I)),\,\mathrm r(e,\beta(I))\big) \;=\; -\big(\mathrm r(e,\beta(I)) - \mathrm r(e,\tau(I))\big)^2
\]
This function is maximized at 0 (when interpretation rank equals intent rank), explaining why the user prefers the interpretation $\beta_2$. To explain the data source's preference for the interpretation $\beta_1$, we define bias function $b$ based on the brand attribute:
$$
b(e) = 
\begin{cases} 
2 & \text{if } e.\text{brand} = \text{'Skullcandy'} \\
0 & \text{otherwise}
\end{cases}
$$
Therefore, the source's utility for a Skullcandy tuple $e$ becomes $u^s_e(\cdot) \;=\; -\Big(\mathrm r(e,\beta(I)) - \big(\mathrm r(e,\tau(I)) + 2\big)\Big)^2$. Thus the source's ideal output rank for Skullcandy tuples is shifted two positions lower (numerically larger) than the intent rank, so it prefers $\beta_1$ over $\beta_2$.
\end{example}



\noindent
\textbf{Partial Common Interest.} 
Although the data source may be biased towards returning some tuples more than other tuples, the results returned by the data source should not be completely independent of the user intent.
In particular, the returned results should have connections to the user intent when the biases are equal. 
More precisely, 
if the position of the tuple $e$ in intent $\tau$ is higher than its position in intent $\tau'$, the data source should place $e$ in the answer to $\tau$ in the same or higher position than the position of $e$ in $\tau'$.


\begin{definition}
\label{definition:supermodular}
Given intent $\tau$ and interpretations $\beta$ and $\beta'$ on an instance $I$ of the schema $\mathcal{R}$ and the data source utility function $U^s$, $U^s$ is \textbf{supermodular} if for all tuples $e$ where 
$\mathrm r(\beta(I),e)$ $<$ $\mathrm r(\beta'(I),e)$, we have  
$u^s(\mathrm r(\tau(I),e),\mathrm r(\beta(I),e),b(e))$ $ - $ $u^s(\mathrm r(\tau(I),e),\mathrm r(\beta'(I),e),b(e))$ is non-increasing in $\mathrm r(\tau,I,e)$. 
\end{definition}

In the definition of supermodularity, we compare the positions of a tuple in different intents, as modifying the position of a tuple in an intent creates a new intent. 
A data source with a supermodular utility function may {\it not} follow the user's preferred relative ranking between different tuples in an intent: it may still be more biased to returning some tuples in positions higher than other tuples.  
Intuitively, the user prefers results that match their preferences the most: the more the user prefers the tuple $e$, the higher position the user would like to see it in the returned result.
Hence, it is natural to assume that $U^r$ is also supermodular, defined similarly to Definition~\ref{definition:supermodular}.
For the rest of the paper, we assume utility functions of the data source and user are supermodular.

\begin{example}
    Consider the supermodular utility functions in Example~\ref{ex:quadratic-bias}. For each tuple, the data source's optimal position in interpretation is $\mathrm r(e,\beta(I))=\mathrm r(e,\tau(I))+2$, while the user prefers interpretations with $\mathrm r(e,\beta(I))=\mathrm r(e,\tau(I))$. Therefore, there is some commonality of interest for each tuple's position since both agents’ optimal interpretation positions are increasing in $\mathrm r(e,\tau(I))$, but also some conflict since the data source’s optimal interpretation position is always $+2$ lower.
\end{example}

\paragraph{Symmetric Bias.}
The data source may be equally biased towards or against the set of tuples of an intent. 
For example, the intent may ask for a ranked list of products produced by a certain provider based on their quality where the data source is biased against products made by that provider. 
In this case, the data source may try to discourage the user by returning and placing products from other providers, which do not belong to the answers of the intent over the current database instance, in high positions of the returned result. 
We show that in these cases the interaction is influential.
\begin{proposition}
\label{prop:symmetry}
Given the interaction setting $(\mathcal{R},\tau,U^r,U^s)$, if $b(e)$ is equal for all tuples $e$, the interaction has influential equilibria.
\end{proposition}

\begin{theorem}
\label{prop:convex-saturation-tuples}
Given interaction setting ($\mathcal{R}$,$\tau$, $U^r$,$U^s$), let $U^s$ be 
\[
U^s(\tau,\beta) \;=
- \sum_{e \in \tau(I) } L\!\big(\mathrm r(e, \tau(I)) - (\mathrm r(e, \beta(I)) + b(e))\big)
\]
over instance $I$ of $\mathcal{R}$, where $L:\mathbb R\to\mathbb R_+$ is a convex function. 
Assume that $U^r$ has the same structure.
If there is an interpretation $\beta'$ such that for every $e$, we have $\beta'(I)$ $=$ $ \arg\min_{\beta(I)} L\!\big(\mathrm r(e, \tau(I))-(\mathrm r(e, \beta(I))+b(e))\big)$, the interaction does {\it not} have any influential equilibrium.
\end{theorem}

Intuitively, Theorem~\ref{prop:convex-saturation-tuples} shows that if the data source's bias is sufficiently large such that there is an interpretation that provides the maximum utility for the data source regardless of the user's desired  intent, no influential equilibrium is possible, i.e., users cannot convince the data source to return other results.
If $L$ is differentiable, we can check the condition in Theorem~\ref{prop:convex-saturation-tuples} efficiently. 

The subsequent result follows directly from Theorem~\ref{prop:convex-saturation-tuples} and recognizes some non-influential interactions efficiently for a large group of utility functions, e.g., quadratic utility functions  where $L(.)$ is $(\text{r}(e, \tau(I)) - (\text{r}(e, \beta(I)) + b(e))\big)^2$.
It is particularly useful when the user is interested in the top-$k$ returned results.

\begin{corollary}
\label{coro:quad-bias-saturation}
Given the interaction $(\mathcal{R},\tau, U^r,U^s)$, let $U^s$ be 
\[
U^s(\tau,\beta) \;=
- \sum_{e \in \tau(I)} L\!\big(\mathrm r(e, \tau(I)) - (\mathrm r(e, \beta(I)) + b(e))\big)
\]
for instance $I$ of $\mathcal{R}$, where $L (v): \mathbb R\to\mathbb R_+$ is convex, non-decreasing in $|v|$, and minimized at $v=0$. 
Assume that $U^r$ has the same structure. 
If for every $b(e)$, $|b(e)|\ge k-\tfrac{3}{2}$, the interaction does not have an influential equilibrium where $k$ is the number of tuples in $\tau(I)$.
\end{corollary}

Given the range of $b(.)$, we can check the condition of Corollary~\ref{coro:quad-bias-saturation} in constant time.

\section{Detecting Trustworthy Answers}
\label{sec:possible-informative-equilibria}
A data source may be partially biased towards returning a certain order of tuples for a given query. For instance, a shopping website may return some products from its own owners, some of which may be irrelevant to the user's query, while still returning some relevant products to keep the user engaged. In this section, we investigate these settings and present algorithms to detect reliable information from data sources' interpretation of users' queries.

\paragraph{Trustworthy Information. } We first define the most basic unit of information in our setting.

\begin{definition}
\label{def:trustworthy_tuple}
Given an intent $\tau$ and an interpretation $\beta$ over an instance $I$ of schema $\mathcal R$,
a returned tuple $e\in\beta(I)\cap\tau(I)$ is \textbf{untrustworthy} if there exists a tuple
$e'\in\tau(I)\setminus\{e\}$ such that $\mathrm r(e',\tau(I))<\mathrm r(e,\tau(I))$ and either
(i) $e'\notin \beta(I)$, or (ii) $\mathrm r(e',\beta(I))>\mathrm r(e,\beta(I))$.
Otherwise, $e$ is \textbf{trustworthy}.
\end{definition}

Intuitively, a tuple is untrustworthy if the data source's interpretation misranks the position of a tuple in the intent. As discussed in Section~\ref{sec:efficiently-detecting-uninformative-setting}, the data source may misrank tuples due to bias and may also return an incomplete result, e.g., only the top-$k$ tuples that maximize its utility. Specifically, a returned tuple $e$ is untrustworthy if it may have displaced some tuple $e'$ that the
intent ranks above it, i.e., $\mathrm r(e',\tau(I))<\mathrm r(e,\tau(I))$. This can happen either because the data source returns $e'$ but
ranks it below $e$ in $\beta(I)$ (a \emph{promotion} of $e$), or because the data source omits $e'$
from the returned list altogether (a \emph{demotion by omission})


\begin{problem}
Given the interaction setting $(\mathcal R,\tau,U^r,U^s)$ and an interpretation result $\beta(I)$, detect all untrustworthy tuples in $\beta(I)$.
\end{problem}


\paragraph{Possible Intent Ranks. }
As discussed in Section~\ref{sec:belief-and-euilibrium}, the data source does not know the user's intent $\tau$, so it maintains a prior over the user’s intent. Specifically, it maintains a prior over the \emph{ranks} of tuples in $\tau(I)$. Concretely, for each tuple $e\in\tau(I)$, we assume the rank $\mathrm r(e,\tau(I))$ is a discrete random variable supported on $[z]=\{1,2,\dots,z\}$ and is uniformly distributed on $[z]$ under the prior.



\paragraph{Indifference Boundary. } As discussed in Section~\ref{sec:efficiently-detecting-uninformative-setting}, the supermodularity of the data source's utility function provides some incentive to place tuples that are in higher positions in the intent in relatively higher positions in the interpretation. Specifically, an interpretation that misranks the relative order of tuples in the intent incurs a larger utility loss for the data source when the intent ranks of these tuples are farther apart. However, the bias of the data source may nevertheless lead to interpretations that misrank these tuple that provide higher utility. Especially when two tuples are sufficiently close in the intent, as the data source's utility loss from misranking them is smaller and may be offset by the bias. To characterize when bias provides sufficient incentive for the data source to misrank tuples, we describe how the data source and user compare their utility from alternative interpretations.

\begin{definition}
\label{def:indifference_boundary}
Given intent $\tau$ and interpretations $\beta,\beta'$ on instance $I$ of the schema $\mathcal{R}$, the user is \textbf{indifferent} between $\beta$
and $\beta'$ if $U^r(\tau,\beta)=U^r(\tau,\beta')$. Given a posterior belief $\Phi$ over intents, the data source is \textbf{indifferent}
between $\beta$ and $\beta'$ at $\Phi$ if $\mathbb E_{\tau\sim\Phi}[U^s(\tau,\beta)]=\mathbb E_{\tau\sim\Phi}[U^s(\tau,\beta')]$.
\end{definition}

Intuitively, if the user receives the same utility for any two distinct interpretations, the user is indifferent between them. However, when the user is not indifferent, meaning they strictly prefer one interpretation over the other, they will strategically select the query that convinces the data source to return results based on the preferred interpretation. Therefore, the condition where the user is indifferent effectively defines the decision boundary in the space of possible intents, or equivalently, possible positions of tuples in the intent, based on which the user submits different queries. The data source’s indifference condition is defined analogously, except that the data source evaluates utilities under its posterior belief $\Phi$ because it does not observe the true intent (Section~\ref{sec:belief-and-euilibrium}). 
Next, we provide the structure of indifference conditions for quadratic utility functions.



\begin{proposition}
\label{proposition:affine_boundary}
Given the interaction setting ($\mathcal{R}$,$\tau$, $U^r$,$U^s$)
suppose the user and data source have the following utility functions for each tuple $e$:
\[
u^r(\mathrm{r}(e,\tau(I)), \mathrm{r}(e,\beta(I))) = -(\mathrm{r}(e,\tau(I)) - \mathrm{r}(e,\beta(I)))^2,
\]
\[
u^s_e(\mathrm{r}(e,\tau(I)), \mathrm{r}(e,\beta(I))) = -(\mathrm{r}(e,\tau(I)) - (\mathrm{r}(e,\beta(I)) + b(e)))^2.
\]
over instance $I$ of $\mathcal{R}$. The user is indifferent between interpretations $\beta,\beta'$, iff $\sum_{e\in E}\gamma_e\,\mathrm r(e,\tau)=\alpha^{r}$, where the normal vector is $\gamma_e=2(\mathrm r(e,\beta(I))-\mathrm r(e,\beta'(I)))$. Let $\Phi$ be a belief over intents and $\Phi_e:=\mathbb E_{\tau\sim\Phi}[\mathrm r(e,\tau(I))]$. The data source is indifferent
at belief $\Phi$ iff
$
\sum_{e\in E}\gamma_e\,\Phi_e=\alpha^{s},
\alpha^{s}:=\alpha^{r}+2\sum_{e\in E} b(e)\big(\mathrm r(e,\beta(I))-\mathrm r(e,\beta'(I))\big).
$
\end{proposition}

The proof follows by equating the utilities for $\beta$ and $\beta'$, then simplifying the resulting expression. Proposition~\ref{proposition:affine_boundary} says that the agent's indifference condition has an affine form. Specifically, a hyperplane in the space of possible intent ranks $\{\mathrm r(e,\tau)\}_{e\in \tau(I)}$ for the user, and a hyperplane in posterior-mean space $\{\Phi_e\}_{e\in \tau(I)}$ for the data source, with the same normal but different intercepts due to the bias. The results extend to any additive, supermodular utility functions where the difference $u(\mathrm{r}(e, \tau(I)), \mathrm{r}(e, \beta(I))) - u(\mathrm{r}(e, \tau(I)), \mathrm{r}(e, \beta'(I)))$ is linear in $\mathrm{r}(e, \tau(I))$.

\paragraph{Misrankings. }
Based on Definition~\ref{def:trustworthy_tuple}, a  returned tuple $e$ is untrustworthy exactly when there exists $e'$ that the intent ranks above $e$
(i.e., $\mathrm r(e',\tau(I))<\mathrm r(e,\tau(I))$) but the interpretation does not place $e'$ above $e$
(either $e'$ is omitted, or it is returned below $e$).
This is a \emph{pairwise} violation: it is witnessed by the pair $(e',e)$.
Thus, even if the data source performs a complex reranking (e.g a sequence of swaps), it suffices to analyze pairwise swaps as
canonical local witnesses that isolate a single potentially inverted pair.

\begin{lemma}
\label{lemma:two_tuple_shift}
In the setting of Proposition~\ref{proposition:affine_boundary}, suppose $\beta,\beta'$ differ only on $e\neq e'$ and satisfy
$\mathrm r(e,\beta(I))-\mathrm r(e,\beta'(I))=-(\mathrm r(e',\beta(I))-\mathrm r(e',\beta'(I)))\neq 0$.
Then the user's indifference condition depends on the difference between the tuples' ranks in the intent:
\begin{equation}
\label{equation:indifference-threshold}
\mathrm r(e,\tau(I))-\mathrm r(e',\tau(I))\;=\;\delta,
\qquad\text{where}\qquad
\delta=\frac{\alpha^{s}}{\gamma_{e}}-(b(e)-b({e'}))
\end{equation}
\end{lemma}

Intuitively, Lemma~\ref{lemma:two_tuple_shift} says that when the two interpretations $\beta$ and $\beta'$ swap the position of the two tuples, the user’s preference between the two interpretations depends only on the \textit{relative positions of the tuples in the intent} and the data source's bias. We refer to $\delta$ from Equation~\ref{equation:indifference-threshold} as the \textbf{indifference threshold}. In the following results, we use this indifference threshold to detect trustworthy information in the interaction.


\begin{proposition}
\label{prop:intents_change_interpretation}
In the setting of Lemma~\ref{lemma:two_tuple_shift}, let $\delta_0=\frac{\alpha^{s}}{\gamma_e}$ denote the indifference threshold when $b(e)=b(e')$
and let $\delta=\delta_0-(b(e)-b(e'))$ denote the bias-shifted indifference threshold. Then the data source interpretation is untrustworthy for $e \in \beta(I)$ if there exists a tuple $e'\in \tau(I)$ such that the tuples rank difference in the intent satisfies $r(e,\tau(I))-\mathrm r(e',\tau(I)) \in\big(\min\{\delta_0,\delta\},\ \max\{\delta_0,\delta\}\big].$
\end{proposition}

\paragraph{Baseline Algorithm. }
Proposition~\ref{prop:intents_change_interpretation} gives a baseline algorithm to detect untrustworthy information in the interaction. For each returned tuple $e\in \beta(I)$ and each tuple in the result of the intent on the schema domain $\mathcal{R}$ such that $e'\notin\beta(I)$ or $\mathrm r(e',\beta(I))>\mathrm r(e,\beta(I))$, the algorithm computes the unbiased and biased indifference thresholds
$\delta_0=\frac{\alpha^s}{\gamma_e}$ and $\delta=\delta_0-(b(e)-b(e'))$. If the relative positions of these tuples in the intent satisfy the condition in the proposition, then the $e\in\beta(I)$ is trustworthy. This algorithm is inefficient since for each tuple in the interpretation, it has to compute its positions in the intent based on the schema and scales with the size of the schema domain.

\begin{theorem}
\label{theorem:closed-form-trustworthy-information}
 Given quadratic utility functions for the user and the data source, interpretation result $\beta(I)$ and $\delta\in \{1,\dots,z-1\}$, a tuple $e\in \beta(I)$ with bias $b(e)$ is untrustworthy if there exists any tuple $e'$ such that $e'\notin \beta(I)$ or
$\mathrm r(e',\beta(I))>\mathrm r(e,\beta(I))$ and its bias $b(e')$ satisfies: 
\begin{equation}
\label{equation:untrustworthy-interval}
b({e'}) \in \big[\,b(e)-g(z,\delta),\ \ b(e)-\max\{g(z,\delta)-1,\ s(z,\delta)\}\,\big).
\end{equation}
where $g(z,\delta)=\frac{-\delta^3+3\delta^2 z+\delta z^2+\delta+z^2-z}{d(z,\delta)}$, $
s(z,\delta)=\frac{(z-\delta)(z-\delta+1)(2\delta+z-1)}{d(z,\delta)}$ and $d(z,\delta)=3\big(z^2-z+\delta(2z+1)-\delta^2\big)>0$ 
\end{theorem}

Based on Theorem~\ref{theorem:closed-form-trustworthy-information} a tuple $e\in \beta(I)$ is untrustworthy whenever there exists a tuple $e'$ (either not in $\beta(I)$ or ranked strictly lower than $e$ in $\beta(I)$) whose bias satisfies the condition in Equation~\ref{equation:untrustworthy-interval} for some $\delta$. Specifically, the theorem provides two requirements for $b(e')$ values. First, $g(z,\delta)$ characterizes which bias gaps make $\delta$ a feasible indifference threshold. Second, $s(z,\delta)$ is the minimum bias difference between the two tuples $b(e)-b(e')$ such that the data source ranks $e$ ahead of $e'$ to maximize its utility.

\begin{algorithm}[t]
\caption{\textsc{Efficiently Detecting Trustworthy Answers}}
\label{alg:efficient-credibilility-filter}
\begin{algorithmic}[1]
\footnotesize
\State \textbf{Input:} interaction setting $(\mathcal R,\tau,U^r,U^s)$, interpretation result $\beta(I)$
\State \textbf{Output:} set $\mathcal{C}$ of \textbf{trustworthy} tuples in $\beta(I)$

\State $\mathcal{C}\gets \varnothing$

\State $\mathcal{S}\gets\varnothing$ 
\State Compute $[b_{\min},b_{\max}]$ based on $b(.)$ and $\mathcal{R}$
\For{$\delta=1$ \textbf{to} $z-1$}
    \State compute $g\gets g(z,\delta)$ and $s\gets s(z,\delta)$ \Comment{ Theorem~\ref{theorem:closed-form-trustworthy-information}}
    \State $b^-_\delta \gets \max\{g-1,\ s\}$,\quad $b^+_\delta \gets g$
    \If{$b^-_\delta < b^+_\delta$}
        \State add $(\delta,b^-_\delta,b^+_\delta)$ to $\mathcal{S}$ \Comment{store feasible $(\delta,b^-_\delta,b^+_\delta)$} 
    \EndIf
\EndFor
\ForAll{$e\in \beta(I)$}
    \State $\mathrm{trustworthy}\gets \textbf{true}$
    \ForAll{$(\delta,b^-_\delta,b^+_\delta)\in\mathcal{S}$}
        \Comment{ Interval for possible $b(e')$: $b(e')\in[a,c)$}
        \State $a \gets b(e)-b^+_\delta$;\quad $c \gets b(e)-b^-_\delta$
        \If{$\max\{a,b_{\min}\} < \min\{c,b_{\max}\}$}
            \State $\mathrm{trustworthy}\gets \textbf{false}$
            \State \textbf{break} \Comment{early exit}
        \EndIf
    \EndFor
    \If{$\mathrm{trustworthy}$}
        \State $\mathcal{C}\gets \mathcal{C}\cup\{e\}$
    \EndIf
\EndFor

\State \Return $\mathcal{C}$
\end{algorithmic}
\end{algorithm}

\paragraph{Efficient Algorithm. }
Algorithm~\ref{alg:efficient-credibilility-filter} implements the condition in Theorem~\ref{theorem:closed-form-trustworthy-information} using only the bias range $[b_{\min},b_{\max}]$ of the data source's bias function $b(.)$ based on the schema domain. 
It first computes $g(z,\delta)$ and $s(z,\delta)$ for all $\delta\in\{1,\dots,z-1\}$ and keeps only those $\delta$ for which the theorem’s interval is nonempty, i.e., $b^-_\delta=\max\{g(z,\delta)-1,\ s(z,\delta)\}< b^+_\delta=g(z,\delta).$
For a returned tuple $e \in \beta(I)$ and fixed $\delta$ the algorithm checks Equation~\ref{equation:untrustworthy-interval}. Since individual values $b(e')$ are unknown, the algorithm certifies $e \in \beta(I)$ as trustworthy whenever this interval has empty intersection with the feasible bias range $[b_{\min},b_{\max}]$ for every $\delta$ retained above. If the intersection is empty for all retained $\delta$, then no feasible bias value can satisfy the theorem’s condition, and $e$ is trustworthy. If the intersection is nonempty for some $\delta$, then $e$ is potentially untrustworthy.

\noindent
\textbf{Complexity.}
Computing $g(z,\delta),s(z,\delta)$ for all $\delta$ and constructing $\mathcal{S}$ takes $O(z)$ time. The per-tuple check evaluates at most $|\mathcal{S}|\le z-1$ overlap tests, each in $O(1)$ time. Thus the total running time is $O(z+k|\mathcal{S}|)\subseteq O(kz)$ for $k=|\mathrm{elm}(\beta(I))|$ tuples.

\section{Influential Strategies}
\label{sec:maximum}

The data source may steer results toward sponsored,
high-margin, or otherwise biased products. As a result, directly expressing an intent, e.g., searching
for a particular low-priced product, may yield a ranked list dominated by promoted alternatives. A
user may, therefore, submit a query that strategically adjusts the ranking of products revealed to the
data source to recover relevant results. However, the effectiveness of users' strategies may depend on their interpretation by the data source and the data source’s degree of bias. In this section, we investigate these settings and
propose algorithms for finding strategies that lead the interaction to an equilibrium and improve the user's utility.

\subsection{Detecting Influential Query}
\label{sec:detecting-credibility}
We first characterize what information about the user’s intent can be credibly communicated to the data source. Building on this, we present an algorithm for finding
a query that leads the interaction to an influential equilibrium.



\begin{problem}
Given the interaction setting $(\mathcal R,\tau,U^r,U^s)$ find a query that leads to an influential equilibrium.
\end{problem}

In Section~\ref{sec:possible-informative-equilibria}, we identified untrustworthy rankings in the data source's interpretation. A user aware of such misrankings may counteract this bias by adjusting the ranking of tuples in their queries. For instance, if a shopping website has an incentive to promote a product $e'$ (or demote $e$), the user's query may rank the product $e$ higher than $e'$ only when $e$ is ranked \emph{sufficiently higher} in the intent. We formalize this notion using relative rank constraints.

\paragraph{Relative rank constraints. } Given an intent $\tau$ over an instance $I$ of schema $\mathcal{R}$, tuples $e,e'\in\tau(I)$, and an
integer $\delta\ge 1$, a relative rank constraint is the inequality $\mathrm r(e',\tau(I))-\mathrm r(e,\tau(I)) \;\ge\; \delta.$ Intuitively, the inequality says that $e$ is at least $\delta$ positions above $e'$ in the intent.

\begin{theorem}
\label{thm:pairwise_delta_closed_form}
Consider the same setting as Theorem~\ref{theorem:closed-form-trustworthy-information}.
For an ordered pair $(e,e')\in\tau(I)$, let $b= b(e)-b(e')$, and let $g(z,\delta)$ be as in Theorem~\ref{theorem:closed-form-trustworthy-information}. Define $\delta^\star(e,e')$ to be the (largest) integer solution $\delta\in\{1,\dots,z-1\}$ to $g(z,\delta)-1 \;<\; b \;\le\; g(z,\delta),$
and set $\delta^\star(e,e')=\bot$ if no such integer solution exists.
If $\delta^\star(e,e')\neq\bot$, then the corresponding relative rank constraint is: $\mathrm r(e',\tau(I))-\mathrm r(e,\tau(I)) \;\ge\; \delta^\star(e,e')$
\end{theorem}

Theorem~\ref{thm:pairwise_delta_closed_form} provides a closed form solution to find the relative rank constraint between tuples. Informally, the theorem says that the data source's larger bias in favor of $e'$ requires the user to provide stronger evidence that $e$ truly outranks $e'$ in the intent, i.e., a larger $\delta^\star$. The following result follows from the theorem.

\begin{proposition}
\label{prop:one_delta_influential}
 If there exists a pair $e,e'$ such that $\delta^\star(e,e')\neq\bot$, then the interaction has an influential equilibrium.
\end{proposition}

\begin{example}
\label{ex:delta_threshold_pair}
Suppose an intent $\tau$ over \texttt{Headphones} instance $I$ places $e_2$ far above $e_3$, e.g., $\mathrm r(e_2,\tau(I))=1$ and $\mathrm r(e_3,\tau(I))=3$. Consider the ordered
pair $(e_2,e_3)$ and assume the data source bias difference is $b=b(e_2)-b(e_3)$ with $b<0$, i.e., the data source is relatively more biased to promote $e_3$. Suppose Theorem~\ref{thm:pairwise_delta_closed_form} yields $\delta^\star(e_2,e_3)=2$. The corresponding $\delta^\star$ ranking
constraint is $\mathrm r(e_3,\tau(I))-\mathrm r(e_2,\tau(I)) \ge 2.$
Since $\mathrm r(e_3,\tau(I))-\mathrm r(e_2,\tau(I))=2$, a user with intent $\tau$ can submit this relative ranking information to influence the data source's interpretation. Now consider a different intent $\tau'$ on the same instance where $e_2$ is only one position above $e_3$, e.g., $\mathrm r(e_2,\tau'(I))=2$ and $\mathrm r(e_3,\tau'(I))=3$. Then $\mathrm r(e_3,\tau'(I))-\mathrm r(e_2,\tau'(I))=1<2$, so the above constraint does not hold under $\tau'$ and a user with intent $\tau'$ would submit the complementary information (i.e., $\mathrm r(e_3,\tau'(I))-\mathrm r(e_2,\tau'(I)) < 2$) instead.
\end{example}

\begin{definition}
\label{def:delta_threshold_query}
Given an intent $\tau$ over an instance $I$ of schema $\mathcal{R}$, a $\boldsymbol{\delta}$-\textbf{query} $q^\delta$ is any conjunction of relative
rank constraints of the form $\mathrm r(e',\tau(I))-\mathrm r(e,\tau(I)) \ge \delta,$ where $\delta\in\mathbb Z$. Equivalently, for some set $\mathcal V\subseteq \tau(I)\times\tau(I)\times\mathbb Z$,
\[
q^\delta \;=\; \bigwedge_{(e,e',\delta)\in\mathcal V}\big(\mathrm r(e',\tau(I))-\mathrm r(e,\tau(I)) \ge \delta\big).
\]
\end{definition}

\begin{lemma}
\label{lem:gap_complement}
For integer-valued ranks,
$\neg\big(\mathrm r(e',\tau(I))-\mathrm r(e,\tau(I)) \ge \delta\big)
\quad\text{iff}\quad
\mathrm r(e,\tau(I))-\mathrm r(e',\tau(I)) \ge 1-\delta.$
\end{lemma}

\paragraph{Algorithm. }
 Algorithm~\ref{alg:detect-delta-relations} constructs the query $q^\delta$ that the input intent $\tau$ can submit to influence the data source's interpretation. As discussed in Section~\ref{sec:efficiently-detecting-uninformative-setting}, user's usually do not know the exact instance $I$, therefore the algorithm uses the information based on the domain of schema $\mathcal{R}$. We therefore run the algorithm over the \emph{rank domain} induced by the
schema and the intent’s ranking attributes. For instance, if $\tau$ has a ranking attribute $A$
(e.g., \texttt{ORDER BY rating}), we take the rank domain to be $\mathbf{dom}(A)$; if $\tau$ ranks by
attributes $A_1,\dots,A_l$, we take the rank domain to be the product
$\mathbf{dom}(A_1)\times\cdots\times\mathbf{dom}(A_l)$. We treat $\tau$ as an order
over the domain of the cartesian product and work with its induced rank function $\mathrm r(\cdot,\tau)$ on this domain.

The algorithm iterates over strictly ordered pairs $e\prec_\tau e'$ in this domain and invokes the subroutine (Algorithm~\ref{alg:pair-constraint-quadratic}) to compute $\delta^\star(e,e')$ based on Theorem~\ref{thm:pairwise_delta_closed_form}. While Theorem~\ref{thm:pairwise_delta_closed_form} provides a closed form for the quadratic case, the approach generalizes to other additive, supermodular utility functions with additive bias. For other supermodular utility functions, one may replace the solution for $g(.)$. To avoid materializing all values $g(.)$, the Algorithm~\ref{alg:pair-constraint-quadratic} evaluates $g(.)$ on demand and finds $\delta^\star$ via binary search over $\delta\in\{1,\dots,z-1\}$. The output query $q^\delta$ is a conjunction of relative rank constraints, where $\mathcal V$ contains exactly one constraint per pair considered. Importantly, $q^\delta$ is constructed to be
\emph{true under the input intent} $\tau$. Intuitively, each pair $(e,e')$ with $\delta^\star(e,e')\neq\bot$, the algorithm adds
either the forward constraint $\mathrm r(e',\tau)-\mathrm r(e,\tau)\ge \delta^\star$ if it holds, or otherwise
adds the logically equivalent complementary constraint
$\mathrm r(e,\tau)-\mathrm r(e',\tau)\ge 1-\delta^\star$
(Lemma~\ref{lem:gap_complement}).


\begin{example}
\label{ex:multiple_rankings_from_delta_query}
Suppose given a query over \texttt{Headphones} instance $I$, Algorithm~\ref{alg:detect-delta-relations} outputs the query $q^\delta \;=\; \big(\mathrm r(e_2,\tau(I))-\mathrm r(e_1,\tau(I))\ge 1\big)\ \wedge\
\big(\mathrm r(e_4,\tau(I))-\mathrm r(e_3,\tau(I))\ge 1\big),$
This conjunction enforces $e_1$ to be ranked above $e_2$ and $e_3$ to be ranked above $e_4$, but
imposes no constraint between the pairs $\{e_1,e_2\}$ and $\{e_3,e_4\}$. Hence $q^\delta$ is
consistent with multiple ranked lists. For instance,
$
q^\delta_1(I): e_1 \prec e_2 \prec e_3 \prec e_4, $ $
 q^\delta_2(I): e_3 \prec e_4 \prec e_1 \prec e_2$ both satisfy $q^\delta$.
\end{example}

\paragraph{Semantics of $q^\delta$. }The output of Algorithm~\ref{alg:detect-delta-relations} is a query $q^\delta$ that is a conjunction
of relative rank constraints. Informally, $q^\delta$ covers a  \textbf{set of rankings}
consistent with its constraints. Depending on the constraints in $q^\delta$, the set of rankings may be
(1) empty, (2) singleton, or (3) contains
multiple rankings. 
To provide intuition, note that ranks lie in the finite set $[z]$. Hence, each inequality expands into a
finite disjunction over differences in ranks. In particular, for any
ordered pair $(e,e')$ and any $\delta\in\{1,\dots,z-1\}$,
$\mathrm r(e',\tau)-\mathrm r(e,\tau) \ge \delta
\;\;\equiv\;\;
\bigvee_{t=\delta}^{z-1}\Big(\mathrm r(e',\tau)-\mathrm r(e,\tau) = t\Big).$
Equivalently, since $\mathrm r(e',\tau)\in[z]$, the constraint can be written as
a disjunction over the possible ranks of $e'$ conditional on the rank of $e$ i.e
$\mathrm r(e',\tau)-\mathrm r(e,\tau) \ge \delta
\equiv
\bigvee_{j=1}^{z}\left(\mathrm r(e,\tau)=j \ \wedge\ \mathrm r(e',\tau)\ge j+\delta\right)$, where terms with $j+\delta>z$ are unsatisfiable For example, if $\delta=2$ and
$\mathrm r(e,\tau)=4$, then any $\mathrm r(e',\tau)\in\{6,7,\dots,z\}$ satisfies the inequality. This explains the three cases. Constraints with $\delta>1$ impose only minimum separation, so they
often leave slack and admit many rankings. If enough constraints use $\delta^\star=1$, they can fix
a total order and yield a singleton. If constraints conflict (e.g., force $j+\delta>z$), the set is
empty.

\begin{algorithm}
\caption{Detect Influential Queries}
\label{alg:detect-delta-relations}
\footnotesize
\begin{algorithmic}[1]
\State \textbf{Input:} interaction setting $(\mathcal R,\tau,U^r,U^s)$
\State \textbf{Output:} query $q^\delta$
\vspace{3pt}
\State $\mathcal V \gets \emptyset$ \Comment{Set of relative rank constraints}
\For{each ordered pair $(e,e') \in \tau$ } \Comment{Iterate in $\tau$-order}
    \State $\ell \gets$ \textsc{Constraint-Quadratic}$(e,e',\tau,b(.))$
    \If{$\ell\neq \bot$}
        \State $\mathcal V \gets \mathcal V \cup \{\ell\}$
    \EndIf
\EndFor
\State $q^\delta \gets \bigwedge_{\ell\in\mathcal V}\ell$ \Comment{Definition~\ref{def:delta_threshold_query}}
\State \Return $q^\delta$
\end{algorithmic}
\end{algorithm}

\begin{algorithm}
\caption{\textsc{Constraint-Quadratic}}
\label{alg:pair-constraint-quadratic}
\begin{algorithmic}[1]
\footnotesize
\State \textbf{Input:} $e,e',\tau,b(.)$
\State \textbf{Output:} relative rank constraint (or $\bot$)
\State $\Delta b \gets b(e)-b(e')$ \Comment{Theorem~\ref{thm:pairwise_delta_closed_form}}
\State $\delta^\star \gets$ binary search over $\delta\in\{1,\dots,z-1\}$ using $g(z,\delta)$
\If{$\delta^\star=\bot$} \State \Return $\bot$ \EndIf 
\If{$\mathrm r(e',\tau)-\mathrm r(e,\tau) \ge \delta^\star$}
    \State \Return $\big(\mathrm r(e',\tau)-\mathrm r(e,\tau) \ge \delta^\star\big)$
\Else
    \State \Return $\big(\mathrm r(e,\tau)-\mathrm r(e',\tau) \ge 1-\delta^\star\big)$
    \Comment{Lemma~\ref{lem:gap_complement}}
\EndIf
\end{algorithmic}
\end{algorithm}

\paragraph{Complexity. } Algorithm~\ref{alg:detect-delta-relations} iterates over ordered pairs of tuples in the results of the intent over the schema domain. In the worst case (when $\tau$ is a total order)
resulting in a search space of size $\binom{m}{2}=O(m^2)$, where \(m = \prod_{i=1}^{l} |\mathbf{dom}(A_i)|\) is the size of the domain of the ranking attributes $A_i.. A_l$ in intent $\tau$. For each pair, the algorithm locates $\delta^\star$ satisfying
$\Delta b\in(g(z,\delta)-1,g(z,\delta)]$. If $\delta^\star$ is found by scanning $\delta\in\{1,\dots,z-1\}$, the per-pair
cost is $O(z)$ and the total runtime is $O(m^2 z)$. If $g(z,\delta)$ is nondecreasing in $\delta$, binary search requires
$O(\log z)$ evaluations of $g (.)$ per pair, yielding $O(m^2\log z)$ total time. The output query $q^\delta$ contains at most
$O(m^2)$ constraints.

\paragraph{Bucketization. } The complexity of Algorithm~\ref{alg:detect-delta-relations} depends on the size of the rank domain induced by the intent’s ranking attributes, which can be large when attributes have high-cardinality
domains. Following standard practice~\cite{jinyang2023biasranking, asudeh2018stableranking}, we therefore bucketize large attribute domains into a small number of groups. For example, continuous attributes such as rating can be discretized into $3$--$4$ buckets. This choice is also motivated by practice.
In many data sources, bias is naturally expressed at the level of \emph{groups} of tuples, e.g., owners' brands, sponsored sellers, or coarse price bands, rather than at the level of individual tuples~\cite{jinyang2023biasranking}. We provide the bucketization details in Section~\ref{sec:experiment-results} and experimentally study how this
choice affects the scalability of the algorithm and the user's utility.

\paragraph{Limitations on query language. }
Supporting queries like $q^\delta$ requires a language that can express constraint-based ranking specifications or a data source that accepts a union of ranking queries. As discussed above, $q^\delta$
denotes the set of rankings consistent with its constraints. This set determines whether
$q^\delta$ can influence the data source's interpretation. If it is empty, $q^\delta$ is uninfluential. If it covers the entire space
of queries, e.g., all permutations of tuples in the user's intent, then $q^\delta$ is again uninfluential. Only when $q^\delta$ rules out some rankings but not all
can it convey nontrivial information and influence the data source's posterior. Then users can submit $q^\delta$, and it leads the interaction to an
influential equilibrium (Proposition~\ref{prop:one_delta_influential}). 
However, our current query language (Section~\ref{sec:data-model-query-language}) restricts users
to single SQL queries whose ranking is determined by an \texttt{ORDER BY} clause. Such queries
necessarily produce a single-ranked list and cannot directly encode relative rank constraints. To bridge this gap, we can select one ranking from the non-trivial set of consistent rankings and $q^\delta$ and submit its SQL encoding. Due to the prevalence of SQL, in the
following sections, we focus on interactions with SQL queries and leave extensions to other query languages
as future work.





\subsection{Improving User Utility}
\label{sec:maximally-influential}
In the previous section, we presented a query that convinces the data source to return some relevant results despite the data source's bias. However, it may still return only a small subset of the tuples relevant to the user's intent. In practice, the user may refine this query by adjusting the ranking further. In this section, we investigate such refinement and provide an algorithm to find a query that improves the user’s expected utility.

\begin{example}
\label{example:credible-communication}
Consider the \texttt{Headphones} instance in Table~1 and let the user's
intent $\tau$ be to retrieve the highest-rated headphones. Suppose the data source returns top-$3$ tuples and has per-tuple utility function $u^s_e
=
\big(\mathrm r(e,\tau(I)) - b(e)\big)\,\mathrm r(e,\beta(I))$, with
bias $b(e)=1.9$. Given a query $q$, let $\Phi_e = \mathbb E[\mathrm r(e,\tau(I)) \mid q]$
denote the posterior mean rank of the tuple $e$. Given the posterior mean, the data source maximizes
$
\mathbb E[u^s_e \mid q]
=
(\Phi_e-1.9)\,\mathrm r(e,\beta(I))$.
Hence the data source asigns $e$ position $\mathrm r(e,\beta(I))=1$ if $\Phi_e<1.9$, and  $\mathrm r(e,\beta(I))=4$ (not returned) if
$\Phi_e>1.9$ to maximize utility. Now consider the query equivalent to the user's intent $q:\texttt{ORDER BY rating DESC}$,
which induces the ranking $e_3 \prec e_2 \prec e_1 \prec e_4$.
Based on theory of order statistics~\cite{David2003Order}, conditioning a uniform prior over the $4!$ possible intent orders given $q$ yields posteriors:
$
\Phi_{e_3}=1,
\Phi_{e_2}=2,
\Phi_{e_1}=3,
\Phi_{e_4}=4
$. Therefore only $e_3$ is returned. Now consider another query $q'$ that results in $q'(I):\ e_3 \sim e_2 \prec e_1 \prec e_4$
Under $q'$, the posterior is 
$
\Phi'_{e_3}=\Phi'_{e_2}=\tfrac{1+2}{2}=1.5,
\Phi'_{e_1}=3,
\Phi'_{e_4}=4.$
Thus both $e_3$ and $e_2$ are returned.
\end{example}

Extending the above example, we begin by introducing a class of query refinements that are useful for improving the returned
results.

\begin{definition}
\label{definition:superquery_rank}
Given queries $q,q'$ over an instance $I$ of schema $\mathcal{R}$, $q'$ is a \textbf{super-rank query} of $q$ if:
\begin{enumerate}
  \item $\mathrm{elm}(q(I))\subseteq\mathrm{elm}(q'(I))$.
  \item For all $e,e'\in q(I)$, if $\mathrm r(e,q(I))=\mathrm r(e',q(I))$, then
  $\mathrm r(e,q'(I))=\mathrm r(e',q'(I))$.
  \item For every $e\in q'(I)\setminus q(I)$ and every
  $e'\in q(I)$, we have $\mathrm r(e,q'(I))>\mathrm r(e',q'(I))$.
\end{enumerate}
We denote the set of all super-rank queries of $q$ by $\mathcal{Q}_{\text{super}}(q)$.
\end{definition}

Intuitively, executing the super-rank query ($q' \in \mathcal{Q}_{\text{super}}$) over some instance may (i) change a sub-list of ranked tuples to a set of tuples with the same position or (ii) preserve some of the orders (iii) append additional (ranked) tuples at the end of the list. Each query is a super-rank of itself.

\begin{example}
\label{example:super-rank}
Given a query $q$ that returns $e_1 \sim e_3 \prec e_4 \prec e_2$ over \texttt{Headphones} instance $I$.
Results of some super-rank queries of $q$ are:
\[\begin{array}{ll}
  \text{(a) set}: & e_1 \sim e_3 \sim e_4 \sim e_2\\
  \text{(b) preserve order}: & e_1 \sim e_3 \prec e_4 \prec e_2\\
  \text{(c) append}: & e_1 \sim e_3 \prec e_4 \prec e_2 \prec e_5\\
\end{array}\]
\end{example}

\paragraph{Super-rank strategy. }
As discussed in Section~\ref{sec:detecting-credibility}, under our query language, the user must
submit a single query that induces a single ranked list. Accordingly, given an intent $\tau$, we run
Algorithm~\ref{alg:detect-delta-relations} to obtain the query $q^\delta$
(Definition~\ref{def:delta_threshold_query}), and then fix a base query $q^{\mathrm{base}}$ by selecting
any executable query that is consistent with $q^\delta$.

\begin{definition}
\label{definition:credible-super-rank-strategy}
A user strategy $P^r(\cdot\mid\cdot)$ is a \textbf{super-rank strategy} if for every intent $\tau$ and
query $q$, we have $P^r(q\mid\tau)>0$ only if $q\in \mathcal Q_{\mathrm{super}}(q^{\mathrm{base}})$.
\end{definition}

Next, we show that the interaction between the user and the data source, when the user leverages a super-rank strategy, always leads to an equilibrium.

\begin{theorem}
\label{theorem:equilibrium}    
Given the interaction setting ($\mathcal{R}$,$\tau$, $U^r$,$U^s$)
suppose the user and data source have the following utility functions for each tuple $e$:
\[
u^r(\mathrm{r}(e,\tau(I)), \mathrm{r}(e,\beta(I))) = -(\mathrm{r}(e,\tau(I)) - \mathrm{r}(e,\beta(I)))^2,
\]
\[
u^s_e(\mathrm{r}(e,\tau(I)), \mathrm{r}(e,\beta(I))) = -(\mathrm{r}(e,\tau(I)) - (\mathrm{r}(e,\beta(I)) + b(e)))^2.
\]
over instance $I$ of $\mathcal{R}$. Then, if the user adopts a super-rank strategy, the interaction has an influential equilibrium.
\end{theorem}

\begin{example}
\label{example:many-common-superrank}
Consider the two super-ranks of $q$ from Example~\ref{example:super-rank}.
$q^{sup}_1(I)$:  $e_1 \sim e_3 \sim e_4 \sim e_2 $ and  
$q^{sup}_2(I)$:  $e_1 \sim e_3 \sim e_4 \prec e_2 $
The query that equally prefers three headphone brands ($q^{sup}_1$) and the one that has a higher preference for specific headphone brands ($q^{sup}_2$) are super-rank queries and both lead to equilibria in the interaction.
\end{example}

The above result indicate that to lead the reasoning to some equilibrium, it is sufficient for the user to choose certain types of strategies.
{\it Nonetheless, users may prefer some equilibria over others.} As illustrated in Example~\ref{example:many-common-superrank}, using some super-rank queries, the user can overstate their preference for some tuples that might otherwise not be returned by the data source. However, a super-rank strategy that maps an intent to a query that places too many tuples in higher positions than their true positions in the intent or adds too many tuples to the intent may reduce the expected posterior of all tuples.
For instance, in Example~\ref{example:credible-communication}, again, assume that the data source has only four headphones.
If the user submits a query that returns all headphones in the data source, the expected posterior of every tuple will be $\frac{5}{2}$.
If the data source returns only tuples where posterior means are above $2.5$, it will not return any tuple. Thus, the user strategy must find the right amount of modification to the intent to offset the bias and maximize extracted information.

\paragraph{Fully Influential Strategy. } In a fully influential equilibrium, the data source returns all tuples relevant to the intent behind the submitted user's query, which is desirable for users:

\begin{definition}
\label{definition:fullyinfluential}
An equilibrium $(P^r(. | .), P^s(. | .)$) is {\bf fully influential} if for every query $q$ and intent $\tau$ such that $P^r(q | \tau)$ $> 0$, $\tau$ is a super-rank of $\beta$ for all $\beta$ where $P^s(\beta | q)$ $>0$.
A user strategy $P^r$ is a {\bf fully influential strategy} if it leads to a fully influential equilibrium.
\end{definition}


To achieve a fully influential equilibrium, the user must convince the data source that all its desired tuples are highly relevant to their intent. 
But, as we have explained in the preceding paragraphs, it might not always be possible to achieve that.
Thus, one might aim for an equilibrium that {\it extracts as much useful information as possible from the data source}.

\begin{definition}
\label{definition:maximalinfluential}
An equilibrium $E_{\max}=(P^r_{\max}(\cdot\mid\cdot),P^s_{\max}(\cdot\mid\cdot))$ is
\textbf{maximally influential} if for every
$E=(P^r(\cdot\mid\cdot),P^s(\cdot\mid\cdot))$ in which $P^r$ is a super-rank strategy, for every query $q$, and every
$\beta_{\max}$ with $P^s_{\max}(\beta_{\max}\mid q)>0$, we have that $\beta_{\max}$ is at least as
influential as every $\beta$ with $P^s(\beta\mid q)>0$.
A user strategy $P^r_{\max}$ is a \textbf{maximally influential strategy} if it leads to a maximally
influential equilibrium.
\end{definition}


Next, we present algorithms for finding the user strategy that led to a maximally influential equilibrium (Definition~\ref{definition:maximalinfluential}).

\begin{problem}
Given the interaction setting $(\mathcal R,\tau,U^r,U^s)$ find a query 
  $q^{\star}$, that leads to a maximally influential equilibrium.

\end{problem}

\paragraph{Baseline Algorithm. }
A brute-force approach is to first get the get $q^{\mathrm{base}}$ derived from $\tau$ and then enumerate all possible super-rank queries $\mathcal{Q}_{\text{super}}(q^{\mathrm{base}})$. Then, for each super-rank query ($q^{sup}$), get the data source’s interpretation ($\beta$), and compute the user's utility \(U^r(q^{sup},\beta)\), finally returning one that maximizes user utility. However, the size of the set $\mathcal{Q}_{\text{super}}$ can be super-exponential (Proposition~\ref{prop:enumerate-complexity}), particularly because super-rank queries can append an arbitrary number of additional tuples.
A simple optimization is to separate the set of super rank queries into two subsets: (i) those that append additional tuples and (ii) those that do not. We can first find the best superrank in the second set and then append additional tuples to it. This will reduce the search space significantly. However, this subset of super-ranks is also exponential, and the problem remains computationally hard. We prove that finding the maximally influential query is, in the general case, NP-hard.

 \begin{proposition}
\label{prop:enumerate-complexity}
Let $m$ be the number of tuples in the results of $q^{base}$ over an instance $I$ of schema $\mathcal{R}$, and let $r$ be the number of tuples that may be appended. The algorithm outputs \emph{super-exponential} number of equilibria and runs in super-exponential time in $r$.
\end{proposition}

\begin{theorem}
\label{thm:superrank-nphard}
The problem of finding the maximally influential strategy is NP-hard in general
\end{theorem}




\paragraph{Merge query. }
Despite the general problem being NP-hard, for additive utility functions, the problem has an optimal substructure that is suited for dynamic programming (DP). A super-rank of $q^{\mathrm{base}}$ is obtained by incremental changing a sub-list of ranked tuples to a set of tuples with the same position.

\begin{definition}
\label{def:merge-query}
Given queries $q,q'$ over an instance $I$ os schema $R$, let
$m$ denote the number of distinct
rank positions in $q(I)$. We call $q'$ a \textbf{merge query} of $q$ if:
\begin{enumerate}
  \item $\mathrm{elm}(q(I))=\mathrm{elm}(q'(I))$
  \item there exist integers $1\le i\le j\le m$ such that, for every $e\in q(I)$,
  \[
  \mathrm r(e,q'(I)) \;=\;
  \begin{cases}
    \mathrm r(e,q(I)) & \text{if }\mathrm r(e,q(I))< i,\\
    i & \text{if } i\le \mathrm r(e,q(I))\le j,\\
    \mathrm r(e,q(I))-(j-i) & \text{if }\mathrm r(e,q(I))> j.
  \end{cases}
  \]
\end{enumerate}

\end{definition}
Intuitively $q'$ is obtained from $q$ by merging the consecutive rank positions $i,i+1,\dots,j$
into a single tied position at rank $i$. We denote the set of merge queries of $q$ by $\mathcal Q_{\mathrm{merge}}(q)$.

\begin{lemma}
\label{lem:run-scores}
Given a query $q$ let $m$ denote the number of rank positions in the results of $q$ over schema $\mathcal{R}$. For each interval $1\le i\le j\le m$, let
$q_{i:j}\in\mathcal Q_{\mathrm{merge}}(q)$ denote the merge query obtained by merging exactly
positions $i,\dots,j$ of $q$. Let $u$ be the
expected per-tuple user utility under $q_{i:j}$. For any merge query $q'\in\mathcal Q_{\mathrm{merge}}(q)$, let $\Pi(q')=\{[i_1,j_1],\dots,[i_t,j_t]\}$
be the unique partition of $\{1,\dots,m\}$ into consecutive intervals that $q'$ merges into ties.
Then the user’s expected utility by submitting $q'$ can be written as
$\sum_{\ell=1}^{t} n_{i_\ell:j_\ell}\,\overline u(i_\ell,j_\ell)\;+\;C,
$
where $n_{i:j}=|\{e\in q: i\le \mathrm r(e,q)\le j\}|$ and $C$ is independent of $q'$.
\end{lemma}

Intuitively, Lemma~\ref{lem:run-scores} says that the expected user utility contributed by merging
any consecutive position interval $[i,j]$ depends only on that interval, and not on how other
positions are merged.

\paragraph{Efficient Algorithm. }
 Algorithm~\ref{alg:max-info-merge-dp} computes a maximally influential query. Given an intent $\tau$, we first run
Algorithm~\ref{alg:detect-delta-relations} to obtain the constraint query $q^\delta$, and then
fix an executable base query $q^{\mathrm{base}}$ by selecting any query whose induced ranked list is
consistent with $q^\delta$ (Definition~\ref{definition:credible-super-rank-strategy}). As in Algorithm~\ref{alg:detect-delta-relations}, we do not assume
access to the realized instance; instead we work over the $m$ rank positions induced by the ranking
attributes $A_1,\dots,A_l$ in $\tau$, where $m=\prod_{i=1}^{l}|\mathbf{dom}(A_i)|$.

A merge query
$q'\in\mathcal Q_{\mathrm{merge}}(q)$ coarsens this ranking by selecting consecutive position
intervals and tying all tuples whose original ranks fall in each selected interval. Equivalently,
each $q'$ corresponds to a unique partition of $\{1,\dots,m\}$ into consecutive intervals.
By Lemma~\ref{lem:run-scores}, the expected user utility decomposes as a sum of interval scores.
Thus, optimizing over $\mathcal Q_{\mathrm{merge}}(q)$ reduces to choosing the partition with maximum
total score. For each interval $[i,j]$, the algorithm computes
$Score(i,j)=n_{i:j}\cdot u,
\qquad
n_{i:j}=\bigl|\{e\in q: i\le \mathrm r(e,q)\le j\}\bigr|$
where $u$ is the expected per-tuple user utility when exactly positions $i,\dots,j$ are merged
(and all other positions remain as in $q$). To compute $u$ we evaluate the
data source’s interpretation to the merge query $q_{i:j}\in\mathcal Q_{\mathrm{merge}}(q)$ using its utility function $\beta_{i:j}\;\in\;\arg\max_{\beta'}\ \mathbb{E}\!\left[U^s(\tau,\beta') \mid q_{i:j}\right]$
However, some utility functions may have multiple maxima, and therefore the $\arg\max$ may not be unique. That is, the data source could be indifferent between several optimal interpretations ($\beta$) that provide the same maximal utility. For such utility functions, we assume a consistent \textit{tie-breaking rule}: if multiple interpretations $\{\beta_1, \beta_2, \dots\}$ yield the same maximal utility, the data source chooses the one most favorable to the user. Since the source does not know the user's true intent $\tau$, it selects the interpretation that is closest to the user's query, i.e, $\beta$ that maximizes $U^r(q, .)$. This is a reasonable assumption since the data source wants to keep the user engaged. 

Finally, given $Score(i,j)$, the algorithm applies the DP recurrence in
Proposition~\ref{thm:dp-recurrence} to compute the optimal partition of $\{1,\dots,m\}$ into consecutive
intervals and then backtracks to construct the corresponding merge query
$q^\star\in\mathcal Q_{\mathrm{merge}}(q)$.

\begin{proposition}
\label{thm:dp-recurrence}
Let $\mathrm{OPT}(j)$ denote the maximum total utility attainable by merging a prefix of
positions $1,\dots,j$. Then
\[
\mathrm{OPT}(0)=0,\qquad
\mathrm{OPT}(j)=\max_{1\le i\le j}\big\{\mathrm{OPT}(i-1)+Score(i,j)\big\}.
\]
\end{proposition}

\begin{algorithm}[H]
\caption{Maximally Influential Query}
\label{alg:max-info-merge-dp}
\footnotesize
\begin{algorithmic}[1]
\State \textbf{Input:} interaction setting $(\mathcal R,\tau,U^r,U^s)$
\State \textbf{Output:} $q^\star\in\mathcal Q_{\mathrm{merge}}(q^\mathrm{base})$
\State $q^\mathrm{base} \gets$ Algorithm~\ref{alg:detect-delta-relations} ($\mathcal R,\tau,U^r,U^s$)
\Function{Score}{$i,j$} \Comment{\textit{Score for merging positions $i..j$ (Lemma~\ref{lem:run-scores})}}
  \State $q_{i:j}\gets$ merge-$[i,j]$ query in $\mathcal Q_{\mathrm{merge}}(q)$
  \State $\beta\in\arg\max_{\beta'}\mathbb E[U^s(\tau,\beta')\mid q_{i:j}]$ 
  \State $n\gets\bigl|\{e\in q: i\le \mathrm r(e,q)\le j\}\bigr|$
  \State $u\gets$ expected per-tuple user utility under $(q_{i:j},\beta)$
  \State \Return $n\cdot u$
\EndFunction
\State $\mathrm{DP}[0]\gets 0$ \Comment{\textit{DP over prefixes (Proposition~\ref{thm:dp-recurrence})}}
\For{$j=1$ to $m$}
  \State $(\mathrm{DP}[j],\mathrm{parent}[j])\gets \max_{1\le i\le j}\{\mathrm{DP}[i-1]+\textsc{Score}(i,j)\}$
\EndFor
\Statex
\State $\Pi^\star\gets$ backtrack intervals from $\mathrm{parent}[\cdot]$
\State $q^\star\gets$ merge query in $\mathcal Q_{\mathrm{merge}}(q)$ that merges exactly intervals in $\Pi^\star$
\State \Return $q^\star$
\end{algorithmic}
\end{algorithm}
\paragraph{Complexity.}
As in Algorithm~\ref{alg:detect-delta-relations}, we work over the rank domain induced by the intent’s
ranking attributes $A_1,\dots,A_l$, whose size is
$m=\prod_{i=1}^{l}|\mathbf{dom}(A_i)|$. Algorithm~\ref{alg:max-info-merge-dp} performs dynamic
programming over the $m$ induced rank positions. The DP considers all intervals $(i,j)$ with
$1\le i\le j\le m$, yielding $O(m^2)$ transitions, and backtracking takes $O(m)$ time. Tabulating the
interval scores $Score(i,j)$ also ranges over $O(m^2)$ intervals. If each $Score(i,j)$ can be computed
in $O(1)$ time (e.g., using cached posterior/utility summaries per interval), the overall running
time is $O(m^2)$; otherwise it is $O(m^2)$ plus the cost of computing all $Score(i,j)$ values.

\begin{figure*}[t]
\centering
\begin{subfigure}{0.199\linewidth}
  \includegraphics[width=\linewidth]{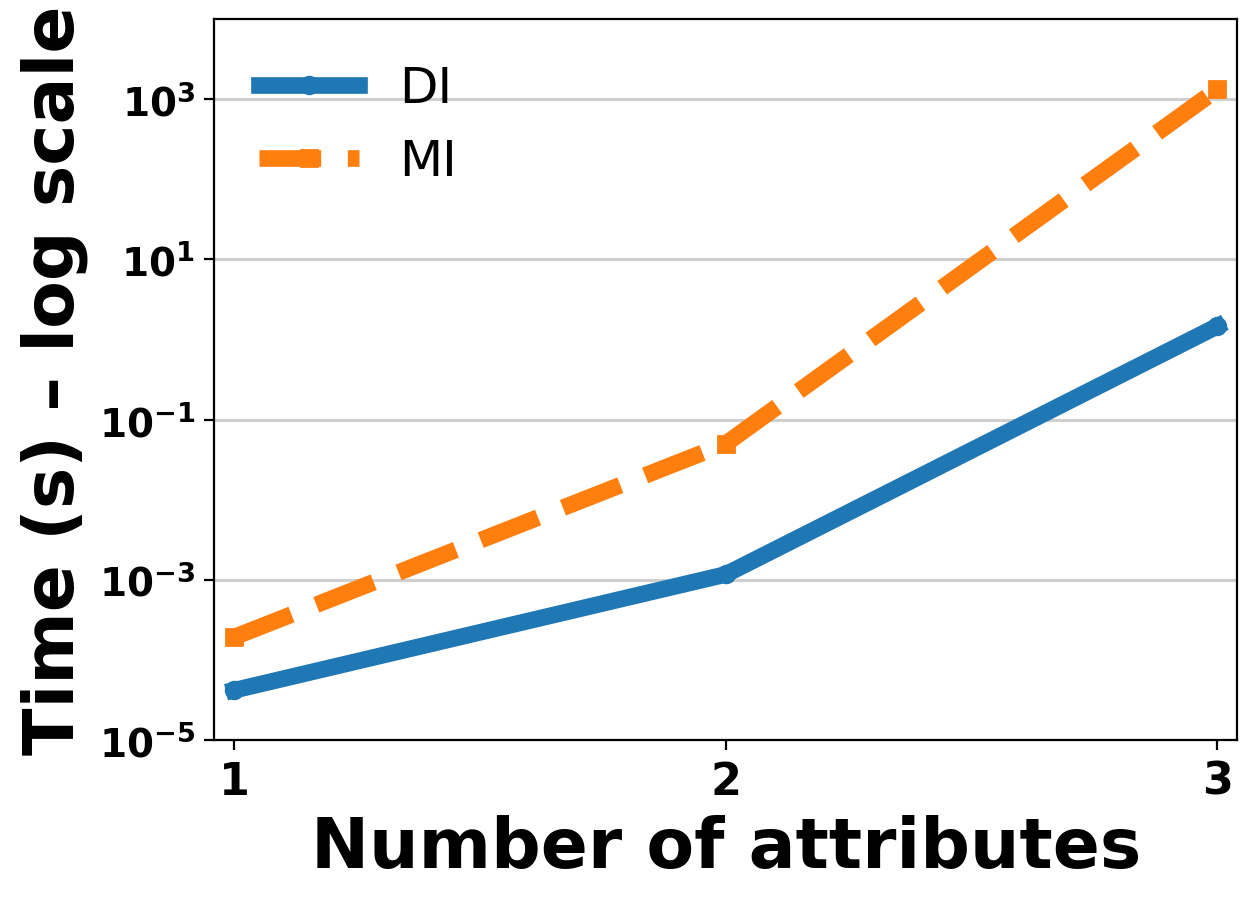}
  \caption{Amazon}
  \label{fig:time_vs_attrs_amazon}
\end{subfigure}\hfill
\begin{subfigure}{0.199\linewidth}
  \includegraphics[width=\linewidth]{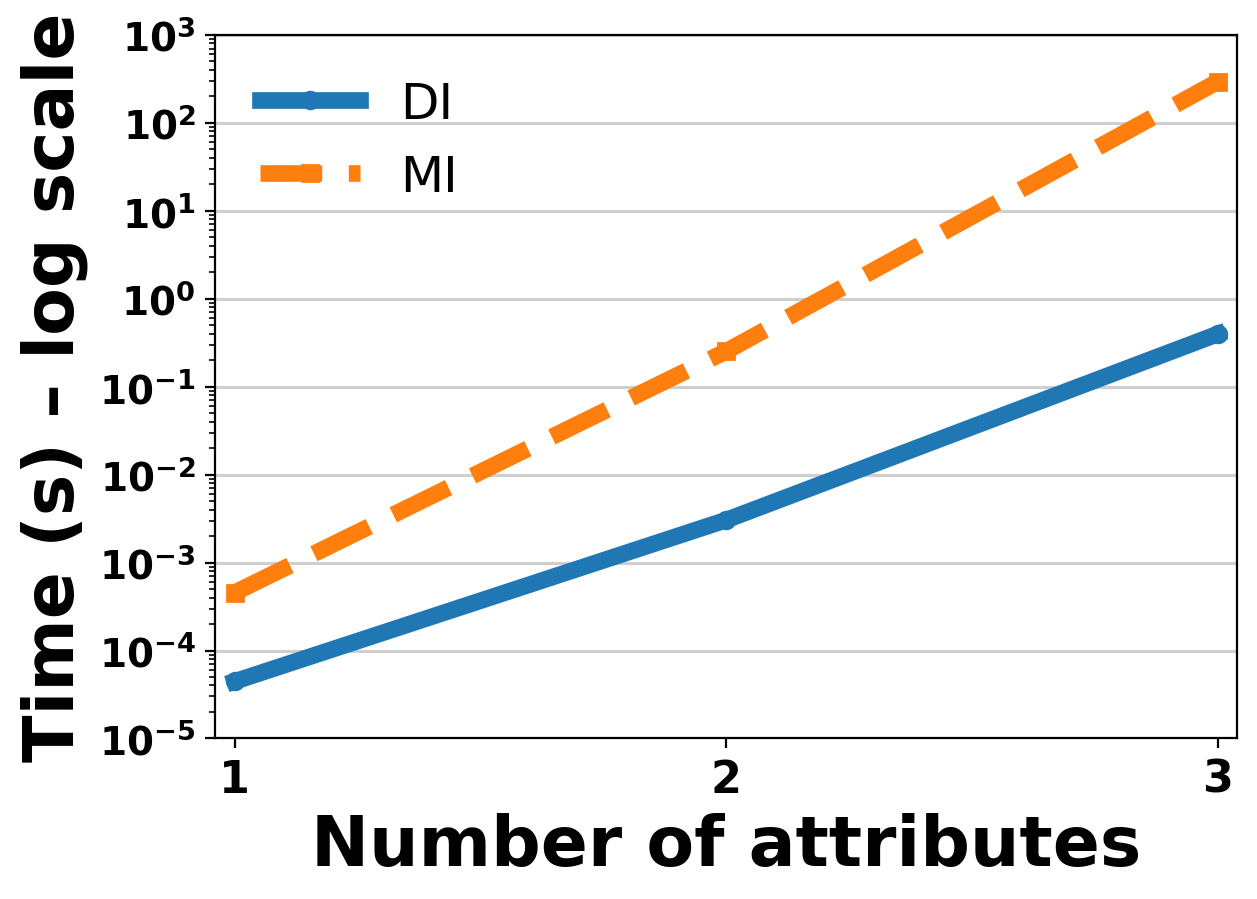}
  \caption{Pricerunner}
\label{fig:time_vs_attrs_pricerunner}
\end{subfigure}\hfill
\begin{subfigure}{0.199\linewidth}
  \includegraphics[width=\linewidth]{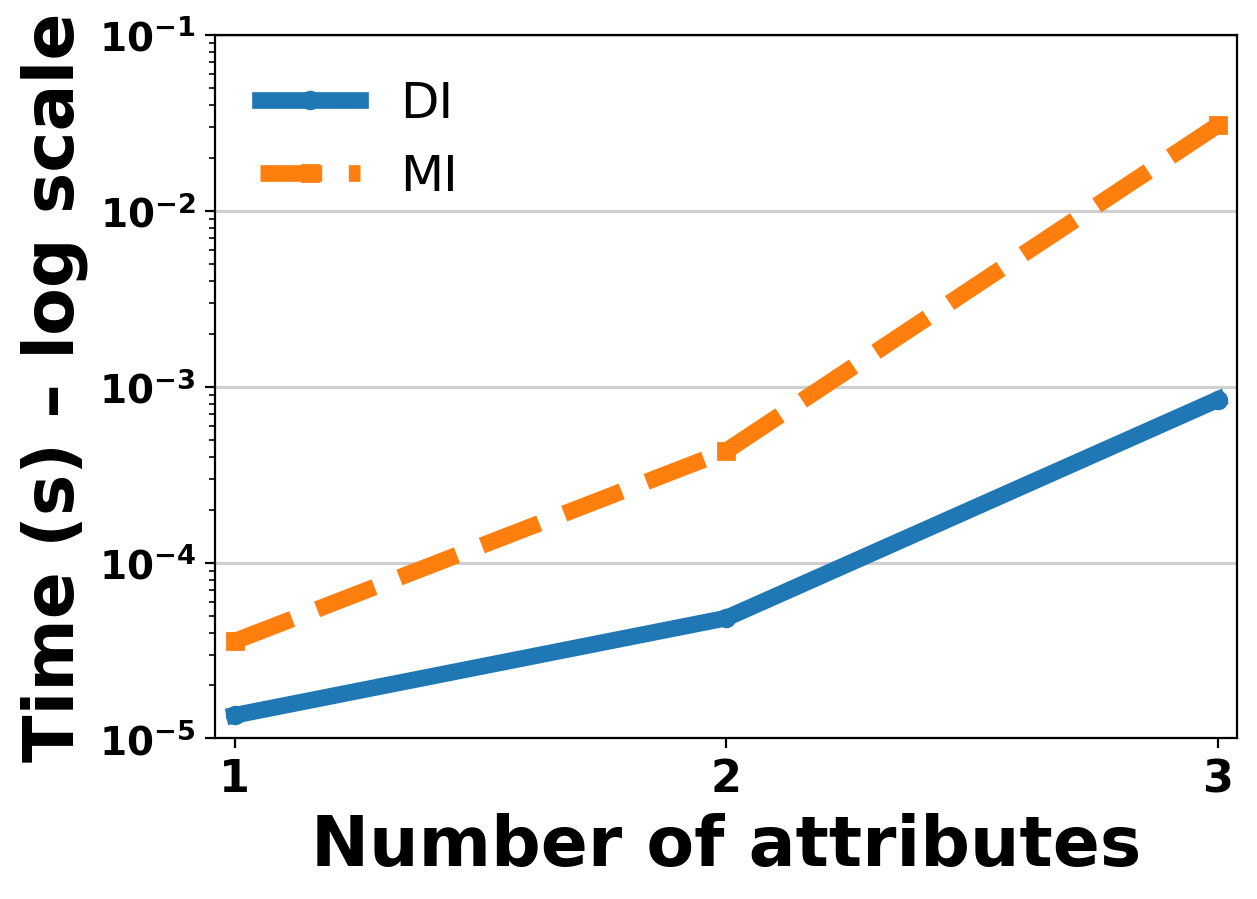}
  \caption{Flights}
  \label{fig:time_vs_attrs_flights}
\end{subfigure}\hfill
\begin{subfigure}{0.199\linewidth}
  \includegraphics[width=\linewidth]{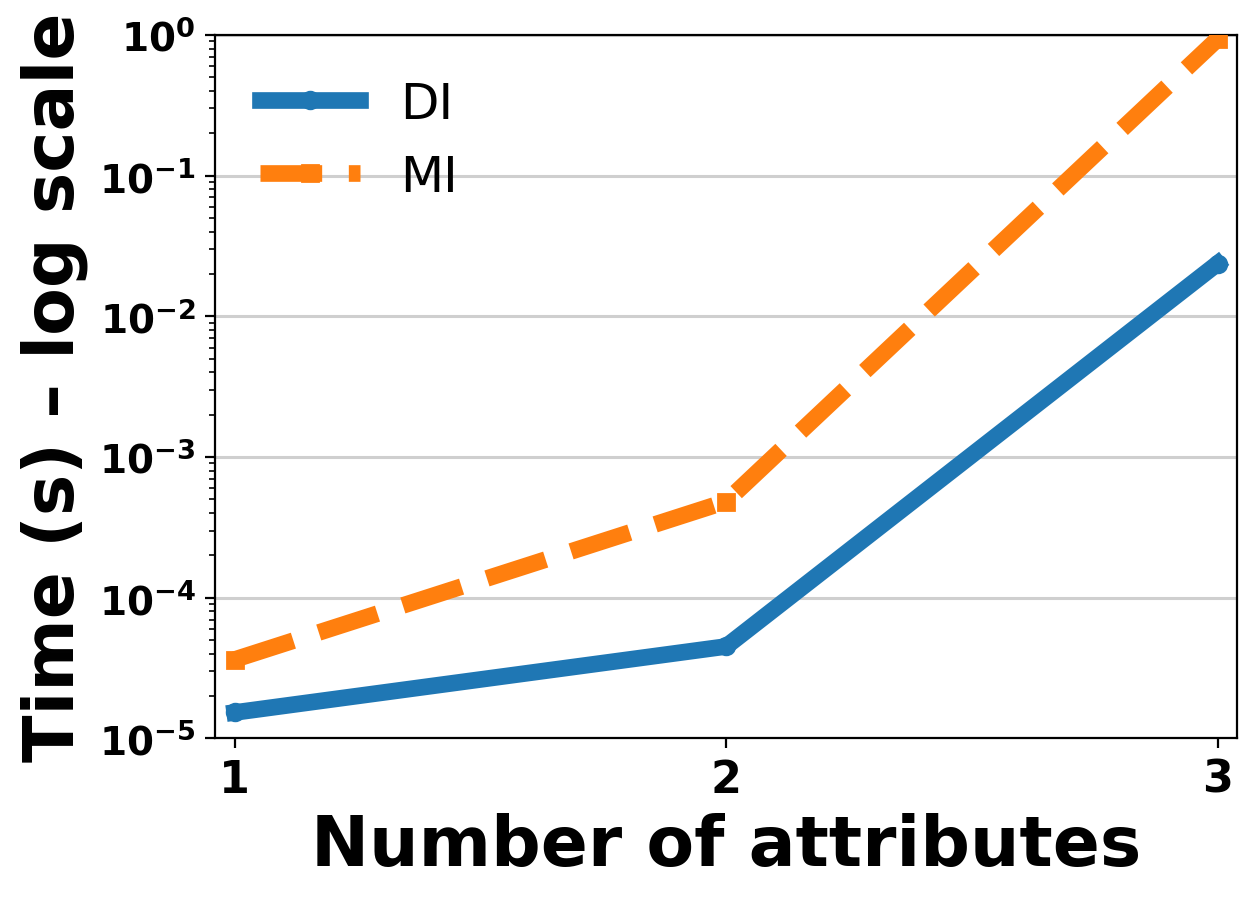}
  \caption{Census}
  \label{fig:time_vs_attrs_census}
\end{subfigure}\hfill
\begin{subfigure}{0.199\linewidth}
  \includegraphics[width=\linewidth]{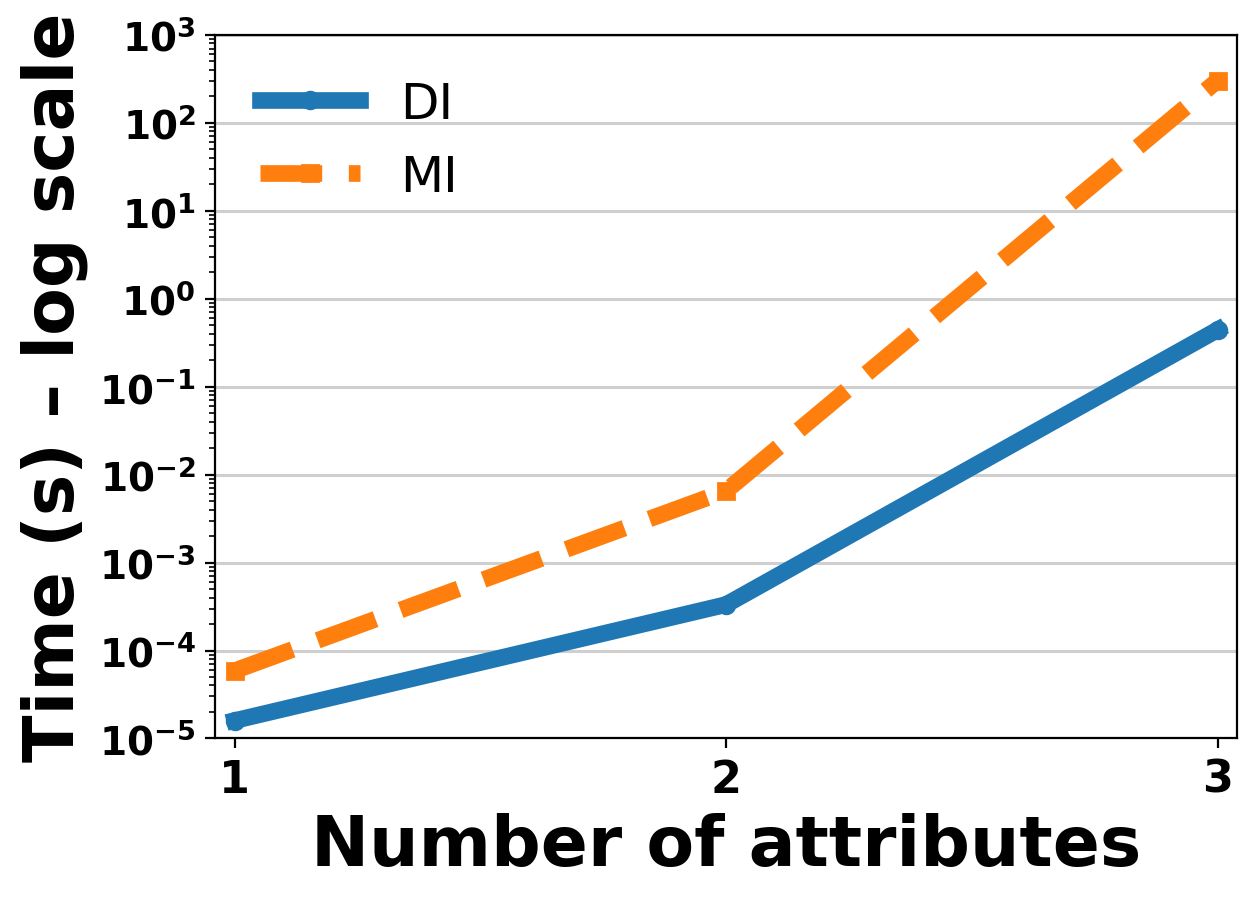}
  \caption{COMPAS}
  \label{fig:time_vs_attrs_compas}
\end{subfigure}

\caption{Impact of number of attributes on execution time of  Detecting Influential (DI) and Maximally Influential (MI)}
\label{fig:time_vs_attrs}
\end{figure*}


\section{Empirical Evaluation}
\label{sec:experiments}

We empirically evaluate our proposed solutions to demonstrate
their usefulness and effectiveness. We describe the experiment setup and then present a quantitative study assessing the efficiency of our algorithm, considering
various datasets, queries, attributes, and the size of the possible domain.
\subsection{Experiment Setup}
\paragraph{Platform. }
All experiments were performed on a macOS Sequoia 15.6 machine with a Appple M1 Pro chip and 16GB RAM.


\paragraph{Datasets. } 
We evaluate our approach on five real-world datasets of varying sizes, domains, and conflict-of-interest scenarios, summarized in Table~\ref{tab:datasets}.

\begin{table}[ht]
\centering
\caption{Summary of datasets used in the experiments.}
\vspace{-3mm}
\label{tab:datasets}
\small 
\setlength{\tabcolsep}{4pt} 
\begin{tabularx}{\columnwidth}{@{}l r r X@{}}
\toprule
\textbf{Dataset} & \textbf{Tuples} & \textbf{Attributes} & \textbf{Ranking Attributes (|A|)} \\
\midrule
Amazon      & 14M  & 11 & rating (41), sales in last month (30), price (29961) \\
PriceRunner & 35K  & 7  & seller (306), product model (12849), category (10) \\
Flights     & 300K & 12 & airlines (6), days from trip (49), price (12157) \\
Census      & 49K  & 12 & race (5), sex (2), education level (16) \\
COMPAS      & 6.8K & 12 & age (65), prior counts (37), violence score (10) \\
\bottomrule
\end{tabularx}
\end{table}

\noindent\textbf{Amazon:} This dataset contains over 14M product listings from Amazon~\cite{amazondataset}. The dataset contains attributes such as product category, price, rating, reviews, number of sales in the last month, and isBestSeller. Since Amazon may have an incentive to offload low-rated and less-selling products to clear inventory, we test for conflict of interest based on rating, sales in last month, and price. 

\noindent\textbf{PriceRunner:} The dataset is collected from PriceRunner\footnote{\url{https://www.pricerunner.com/}}, a popular product comparison platform~\cite{pricerunnerdataset}. It includes 35,311 product offers from 10 categories (e.g., Mobile Phones) provided by different sellers. We test for conflict of interest based on attributes: seller, product model, and product category. For instance, the system may have an incentive to favor products from certain sellers. 

\noindent\textbf{Flights:} The dataset was scraped from EaseMyTrip\footnote{\url{https://www.easymytrip.us/}}, a flight booking service provider~\cite{flights}. It contains 300,261 distinct flight bookings from February 11th to March 31st, 2022. We test for conflict of interest based on attributes: airlines, days from trip, and price. For instance, the booking service may be biased towards expensive flights from certain airlines. 

\noindent
We also consider popular datasets used in prior work related to bias~\cite{jinyang2023biasranking,li2023queyrefinment}. 

\noindent\textbf{Census:} The US Census dataset comprises 49,000 tuples~\cite{adult_2}. We test for conflict of interest based on race, sex, and education level. For example, the system may favor individuals with a higher education level (e.g., PhD). 

\noindent\textbf{COMPAS:} The COMPAS dataset was collected and published by ProPublica\footnote{\url{https://www.propublica.org/article/how-we-analyzed-the-compas-recidivism-algorithm}}. It contains the demographics, violence scores produced by the COMPAS software, and criminal offense information for 6,889 individuals. We test for conflict of interest based on age, prior offense counts, and violence decile score.

\paragraph{Queries. }
For each dataset, we generated a diverse set of SQL queries of the form \texttt{SELECT A\_1, ..., A\_l WHERE Sel(q) ORDER BY A\_1, ..., A\_l} where \texttt{Sel(q)} is a random conjunction of up to three non-key predicates and $l \in \{1,2,3\}$ varies the number of ranking attributes. Different subsets of filtering attributes were randomly selected for the \texttt{WHERE} clause. In the query projection, we include the same set of ranking attributes for simplicity, since our focus is on the conflict based on the ranking attributes. For consistency, we assume that the ranking order in the intent is decreasing for numeric attributes (e.g., {\tt ORDER BY Rating DESC}) and use the inherent order in the original dataset for categorical attributes. 

\paragraph{Utility Function. }
We use the quadratic utility functions throughout our experiments. As discussed in Section~\ref{sec:efficiently-detecting-uninformative-setting}, the utility function of the data source might vary between tuples or sets of tuples. To vary the functions $u_{e}^{s}(.,.)$ and $u_{e'}^{s}(.,.)$ in Equation 1 while maintaining consistency across different datasets, we randomly assign bias values in $[0,1]$ and scale them based on the domain of different datasets. All our reported results are averaged over 5 runs.

\paragraph{Bucketization }
As discussed in Section~\ref{sec:efficiently-detecting-uninformative-setting}, users usually do not know the exact instance $I$, therefore our algorithms reason over the domain of schema $\mathcal{R}$. Therefore, attributes with high-cardinality domains, such as price or product model, impact the scalability of our algorithms. Furthermore, a data source's bias often applies to groups of items (e.g., expensive vs. cheap) rather than specific values, which motivates grouping domain values into a smaller number of bins.  We explore a range of ranking attributes with different domain sizes (Table~\ref{tab:datasets}). Following standard practice~\cite{jinyang2023biasranking}, continuous attributes like age were bucketized into 3-4 groups based on their domain and values. For high-cardinality categorical attributes, we adopted another common approach: keeping the top-$k$ most frequent values and mapping the rest to a single other bin~\cite{asudeh2018stableranking}. For instance, in the PriceRunner dataset, we set $k=1$ and defined bins over the top-10 possible domain values of sellers and models for each category. For datasets with small attribute domains like Census, no bucketization was necessary. We note that the selection of bucketization may affect the credibility of queries and the running time of our algorithms; therefore, we study this effect in Section~\ref{sec:experiment-results}.


\begin{figure}[htbp]
    \centering
    \vspace{-3mm}\includegraphics[width=1\linewidth]{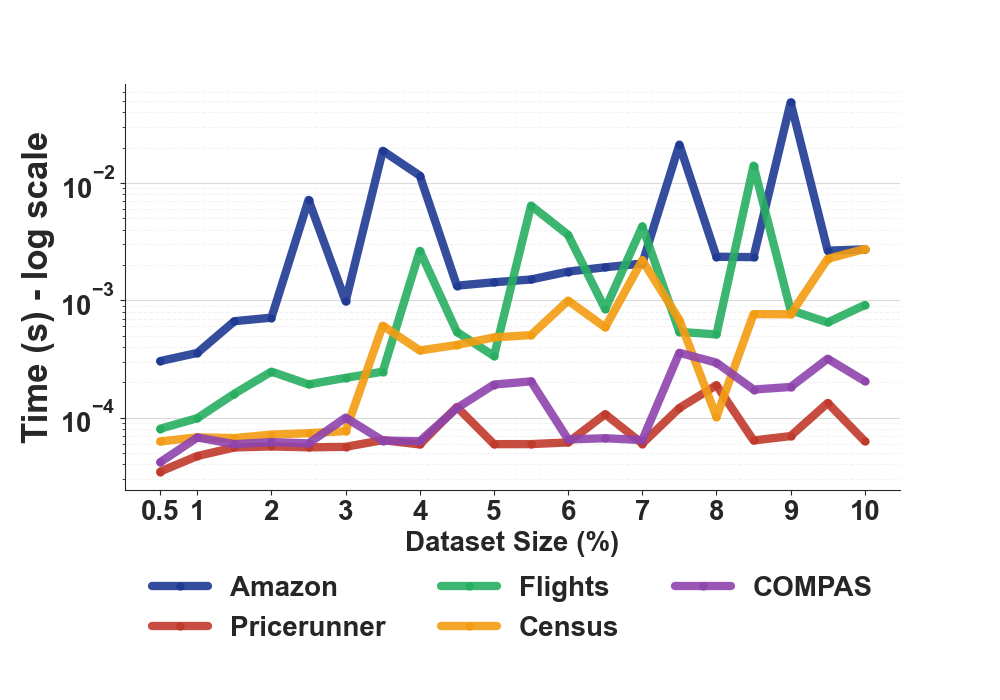} 
    \vspace{-10mm} 
    \caption{ Impact of $z$ on time to detect credible answers}
    \label{fig:all_dataset_credible}
\end{figure}

\begin{figure}[t]
    \centering

    \begin{subfigure}{0.48\columnwidth}
        \includegraphics[width=\linewidth]{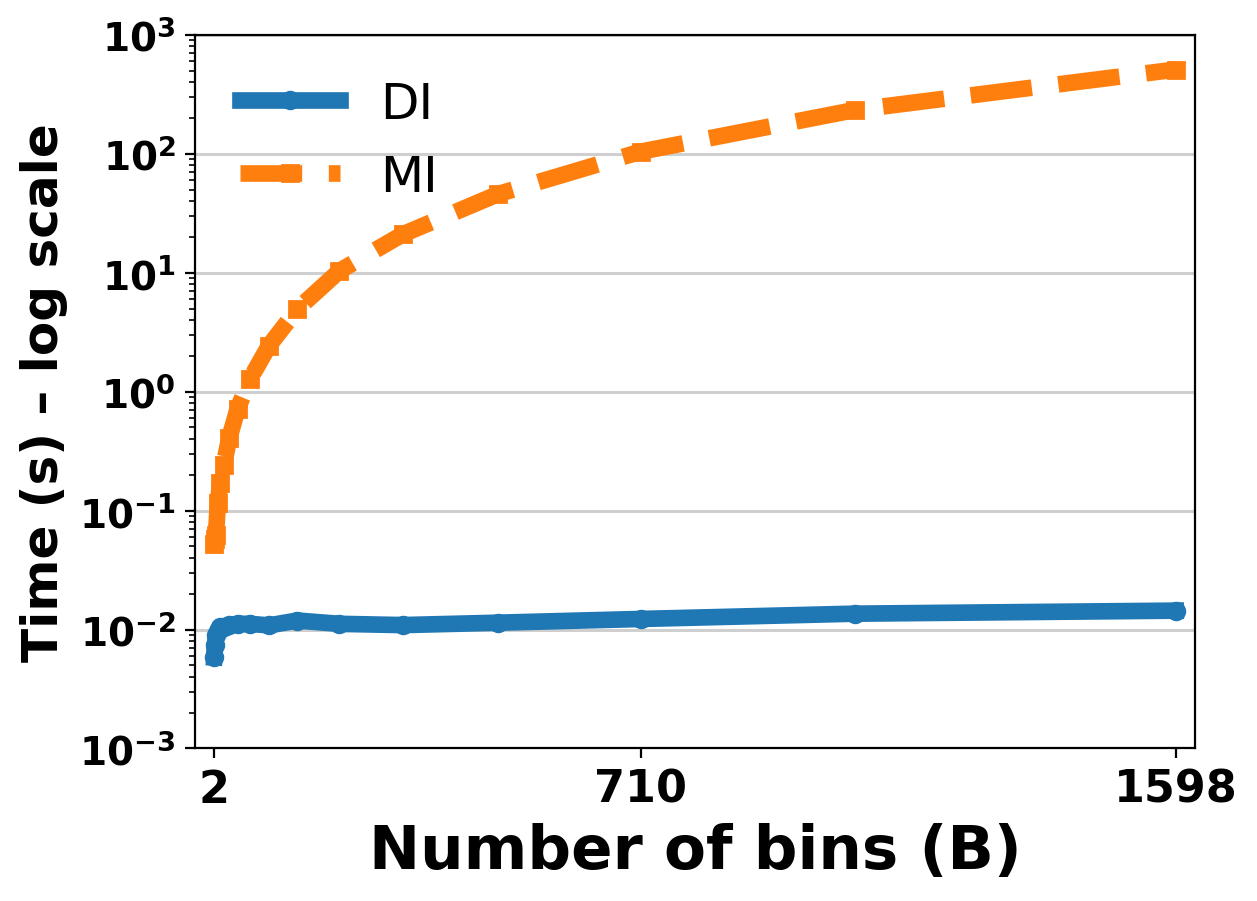}
        \caption{Time (Amazon)}
        \label{fig:time-vs-bins-amazon}
    \end{subfigure}\hfill 
    \begin{subfigure}{0.48\columnwidth}
        \includegraphics[width=\linewidth]{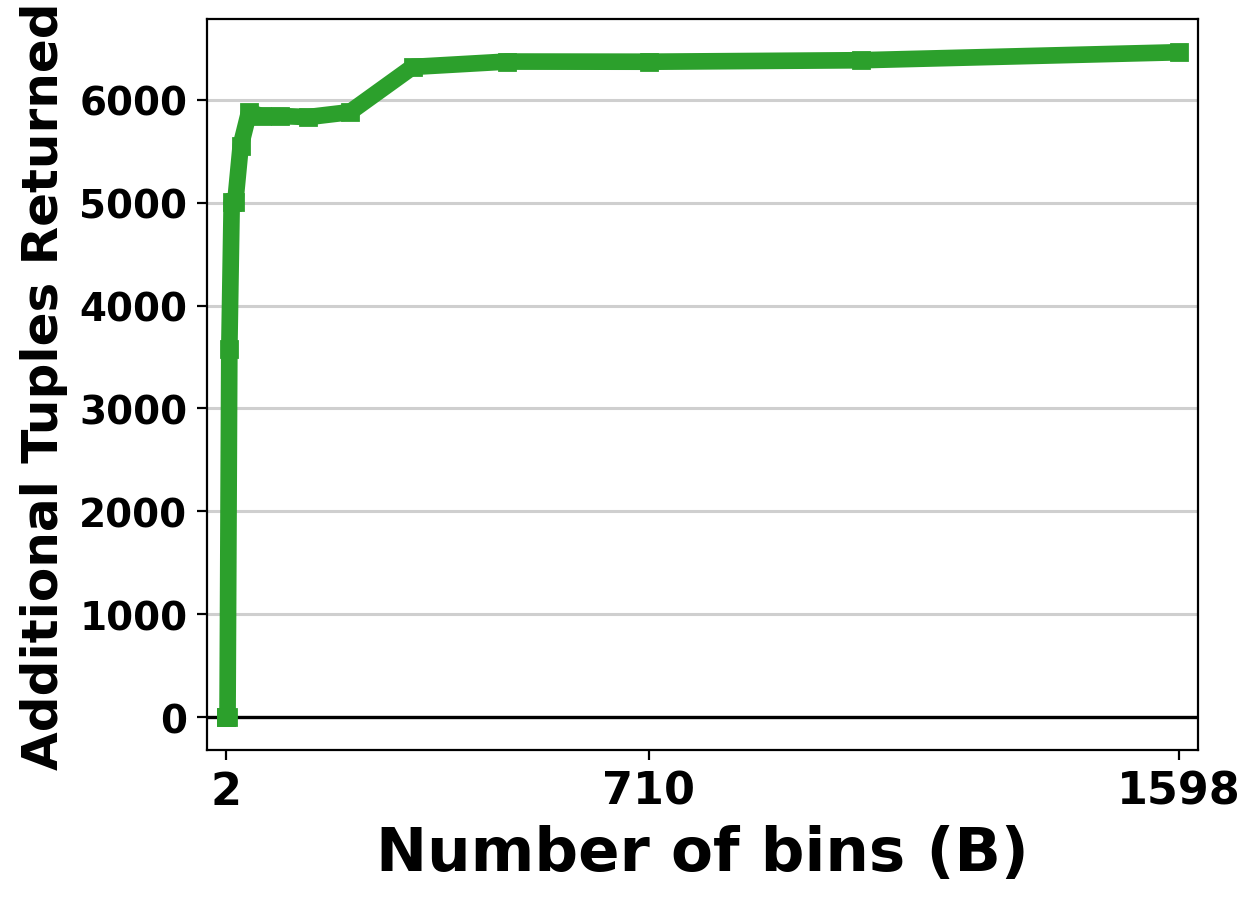}
        \caption{Utility (Amazon)}
        \label{fig:utility-vs-bins-amazon}
    \end{subfigure}

    \vspace{3mm} 

    \begin{subfigure}{0.48\columnwidth}
        \includegraphics[width=\linewidth]{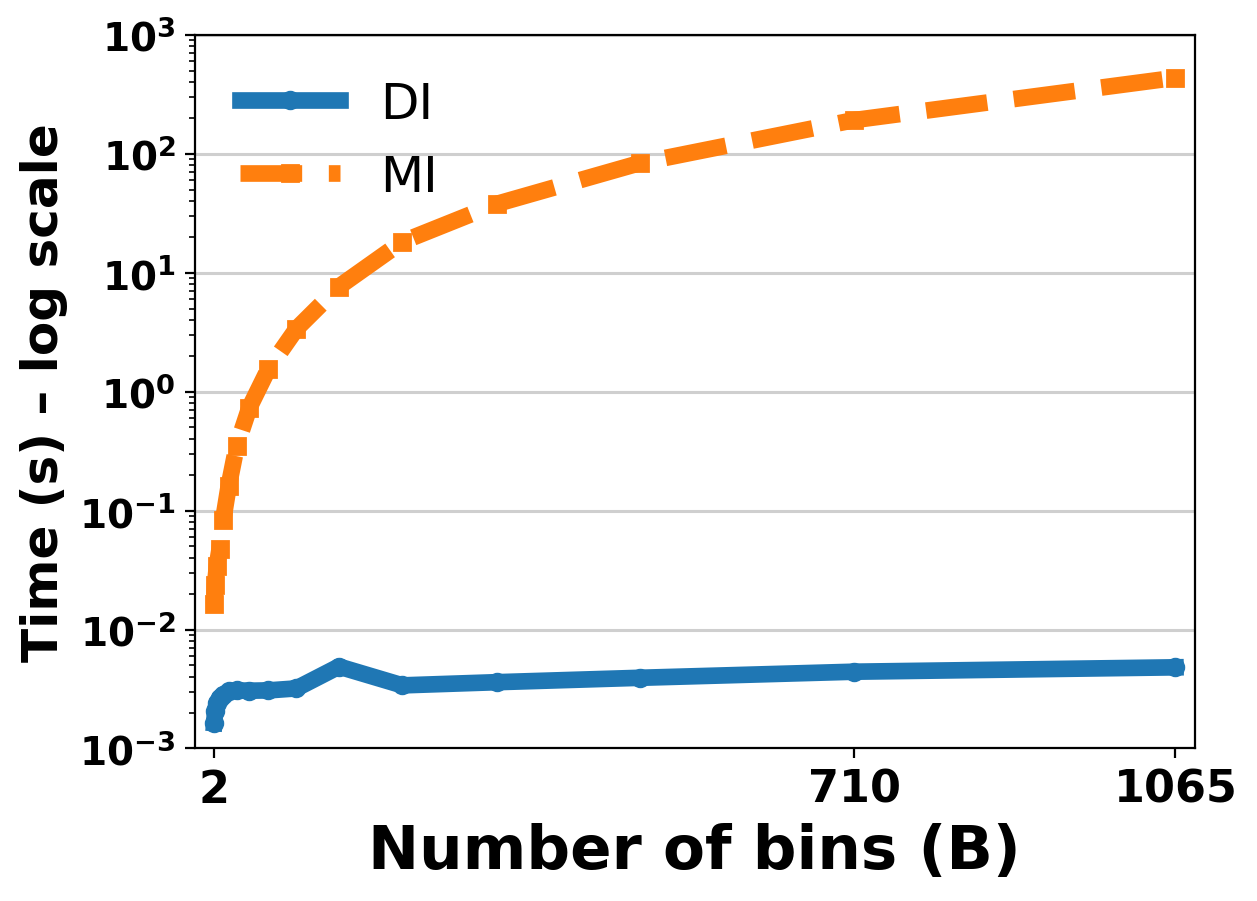}
        \caption{Time (PriceRunner)}
        \label{fig:time-vs-bins-pricerunner}
    \end{subfigure}\hfill
    \begin{subfigure}{0.48\columnwidth}
        \includegraphics[width=\linewidth]{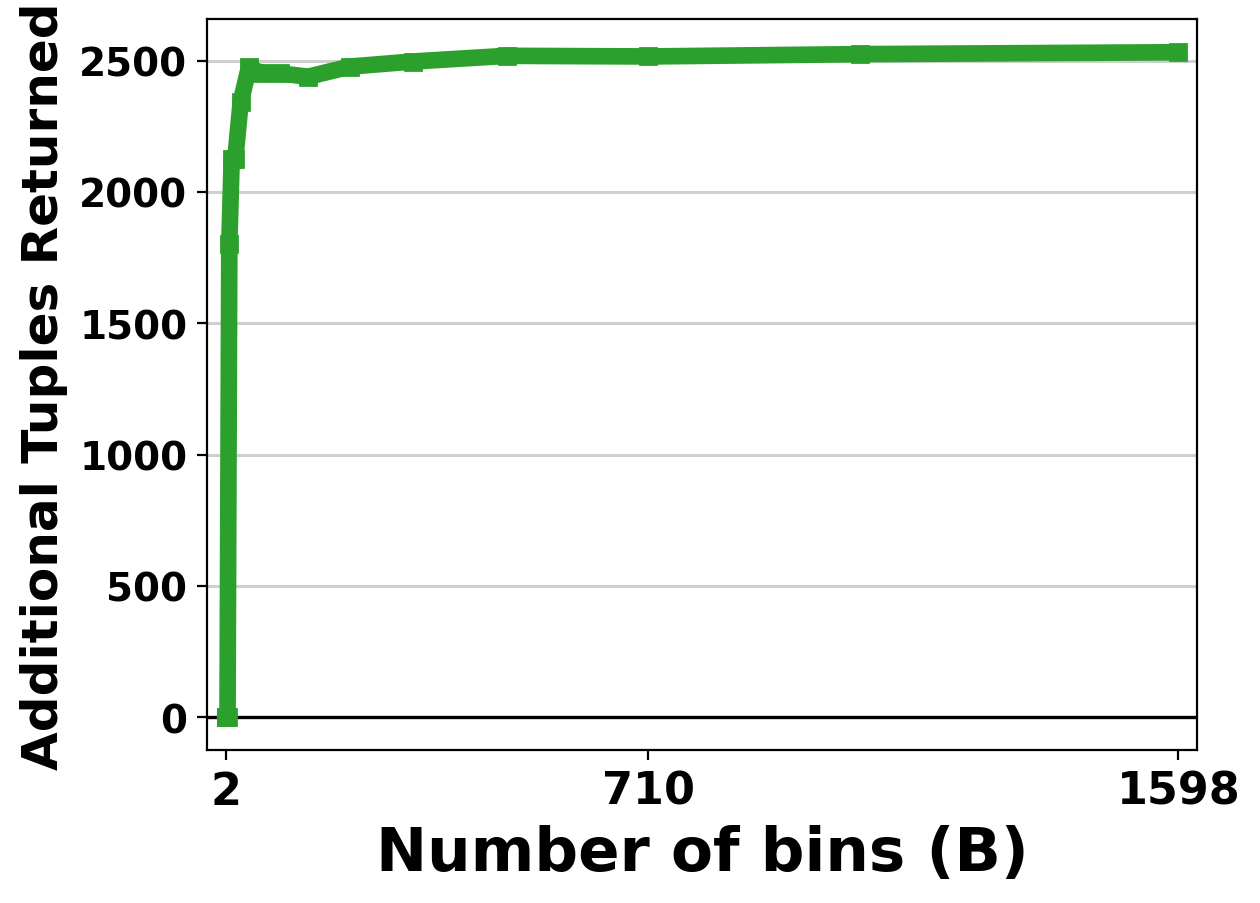}
        \caption{Utility (PriceRunner)}
        \label{fig:utility-vs-bins-pricerunner}
    \end{subfigure}

    \vspace{3mm} 

    \begin{subfigure}{0.48\columnwidth}
        \includegraphics[width=\linewidth]{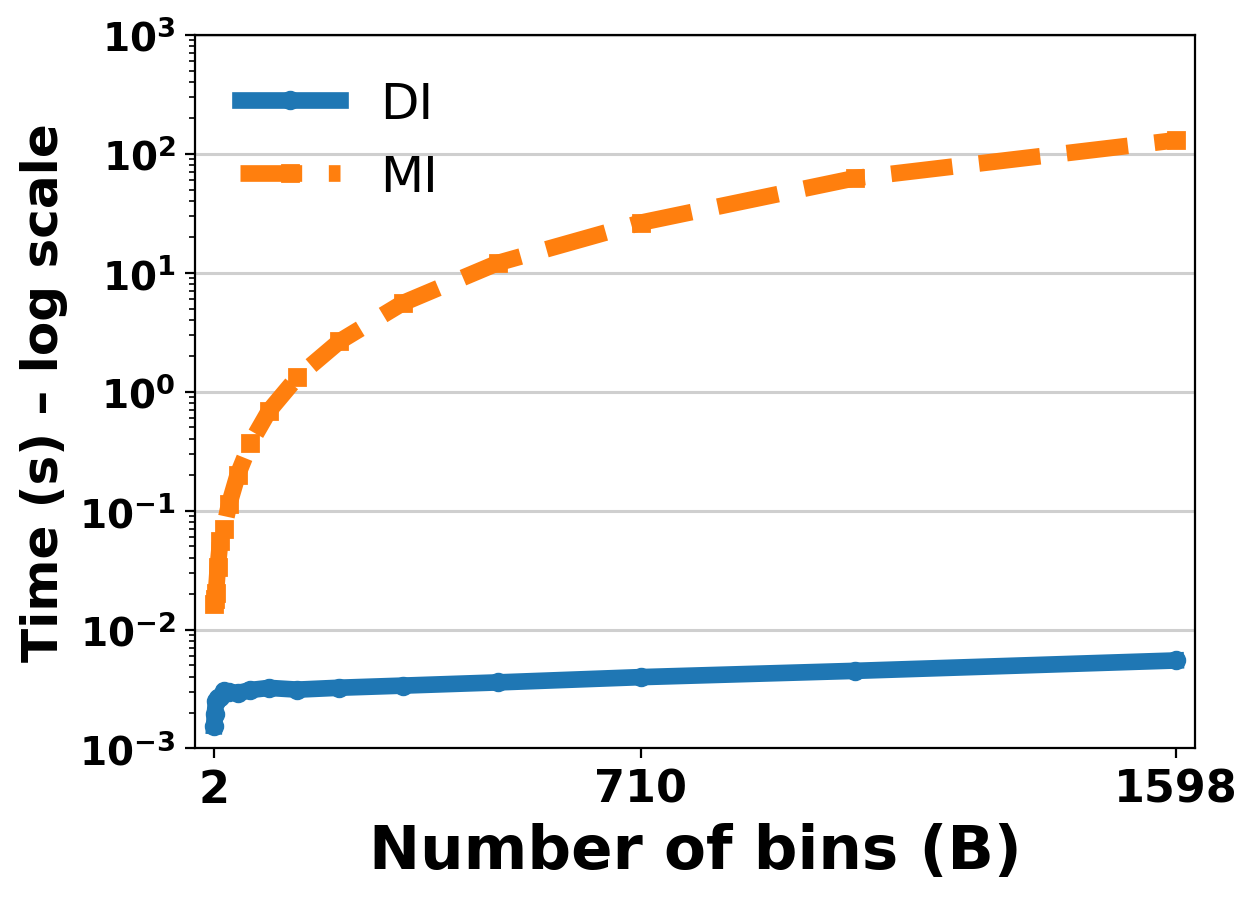}
        \caption{Time (Flights)}
        \label{fig:time-vs-bins-flights}
    \end{subfigure}\hfill
    \begin{subfigure}{0.48\columnwidth}
        \includegraphics[width=\linewidth]{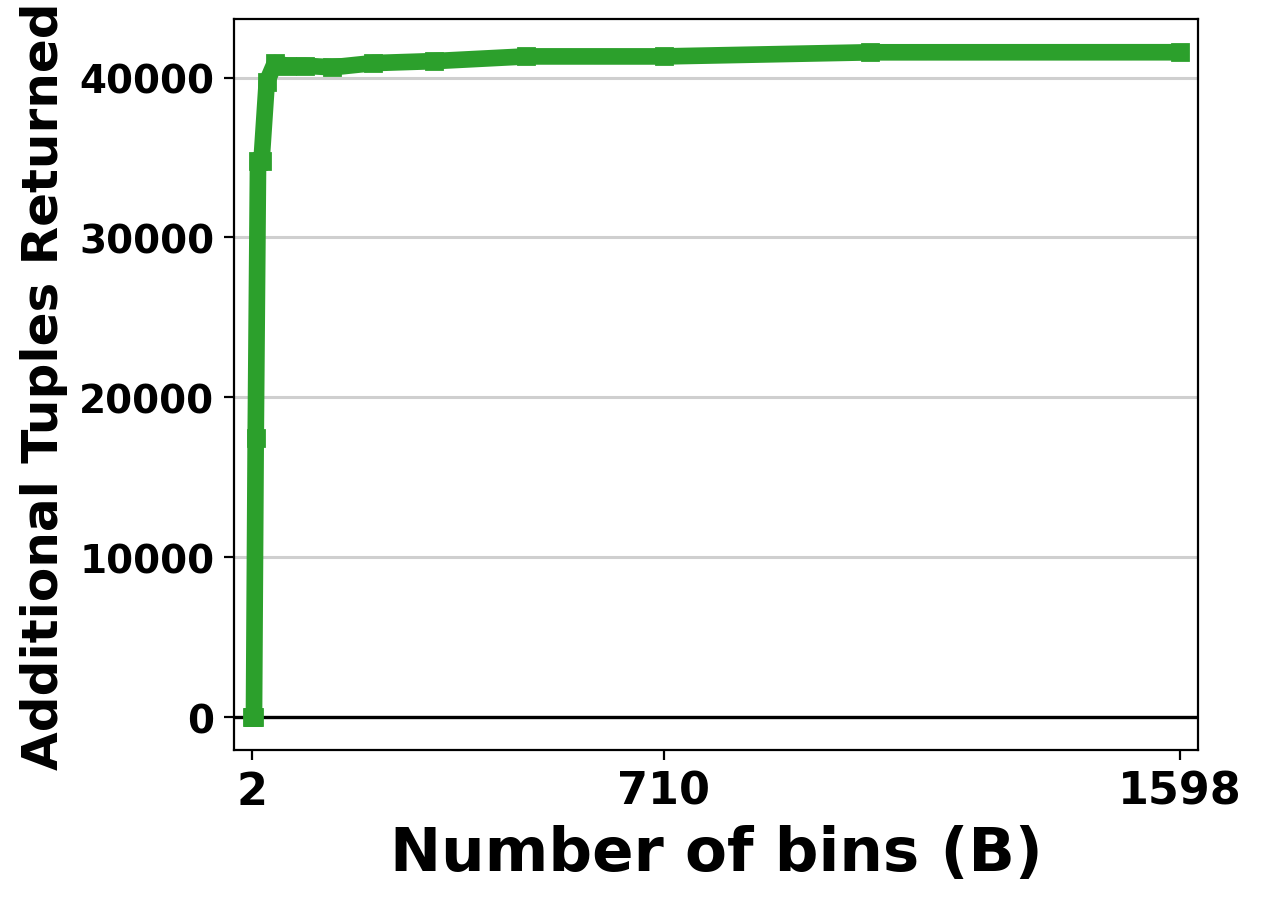}
        \caption{Utility (Flights)}
        \label{fig:utility-vs-bins-flights}
    \end{subfigure}

    \begin{subfigure}{0.48\columnwidth}
        \includegraphics[width=\linewidth]{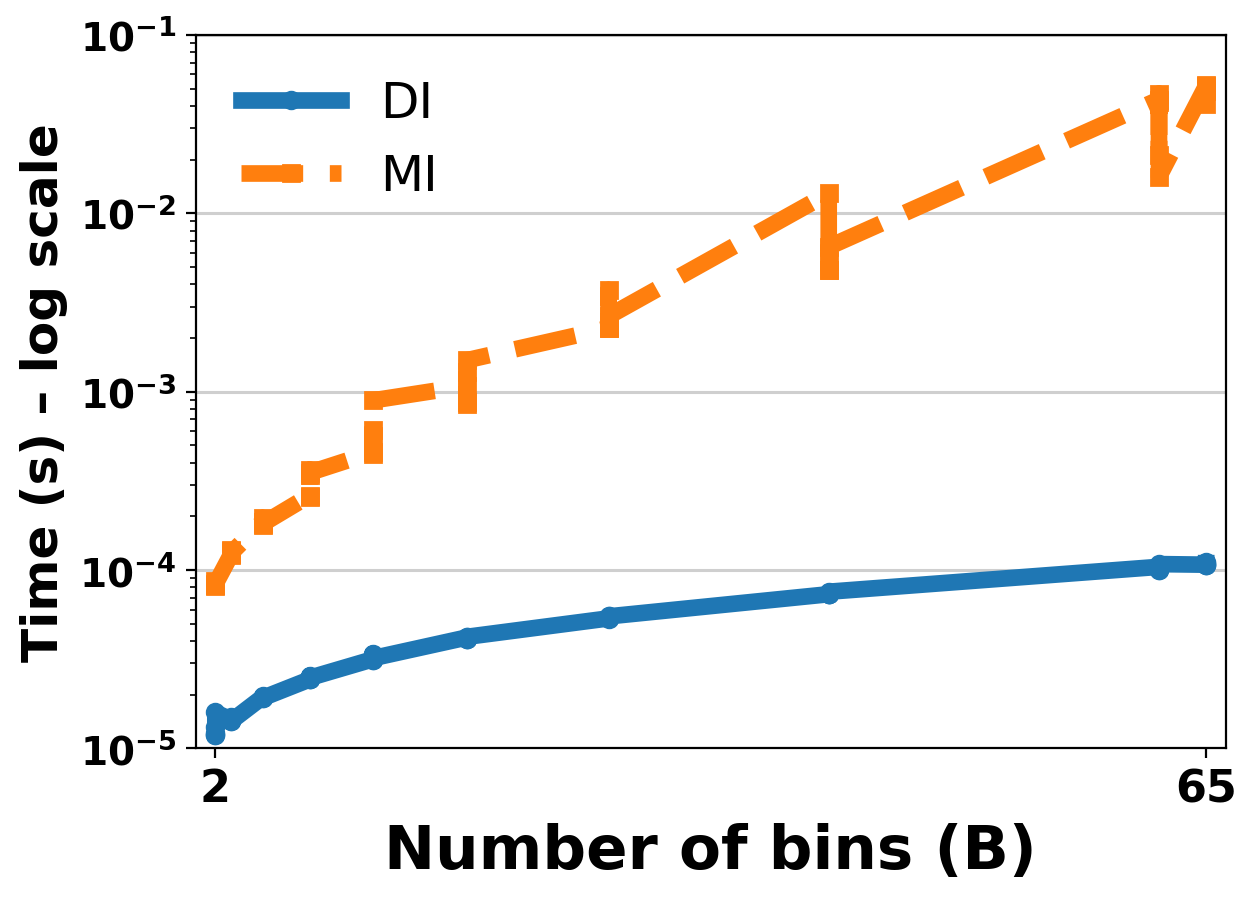}
        \caption{Time (COMPAS)}
        \label{fig:time-vs-bins-compas}
    \end{subfigure}\hfill
    \begin{subfigure}{0.48\columnwidth}
        \includegraphics[width=\linewidth]{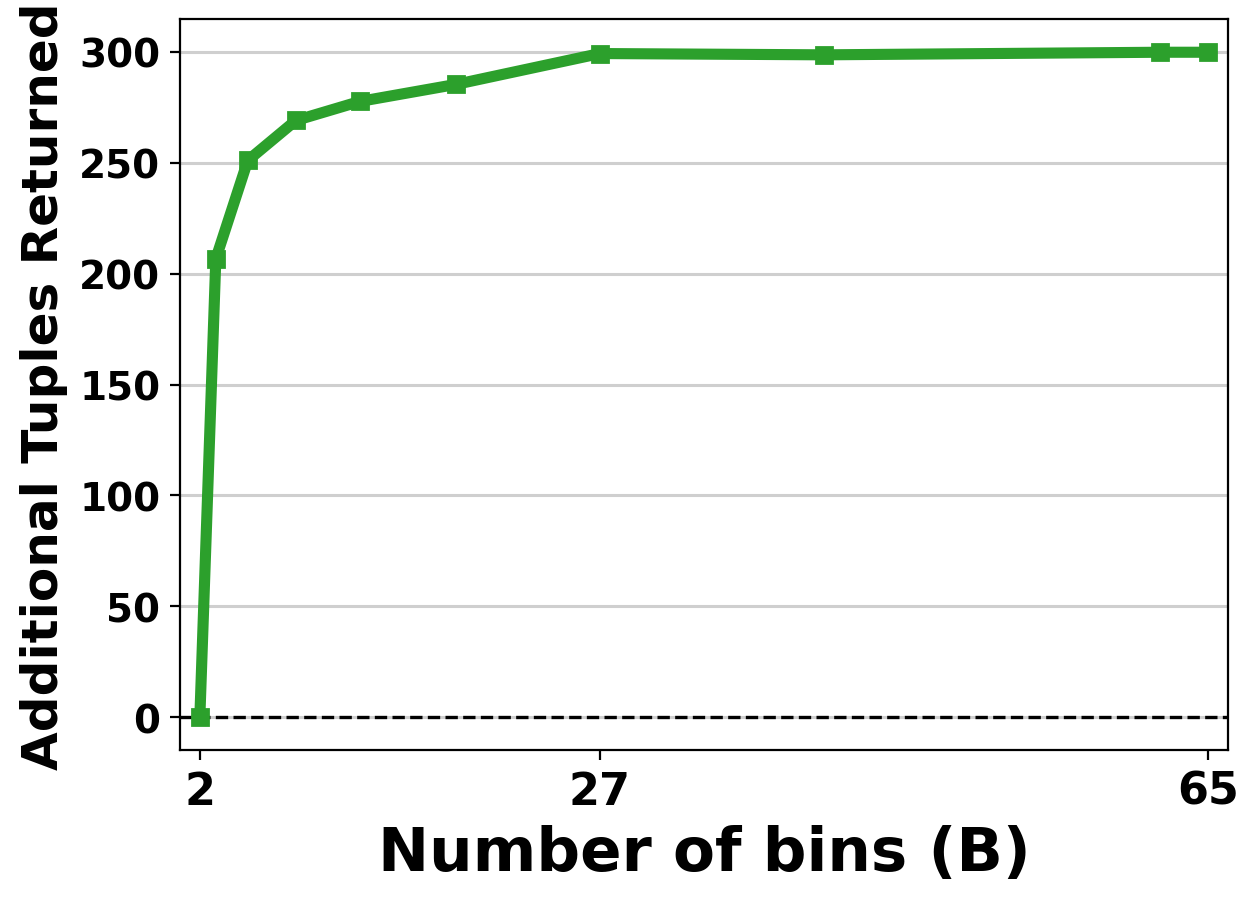}
        \caption{Utility (COMPAS)}
        \label{fig:utility-vs-bins-compas}
    \end{subfigure}

    \caption{Impact of bucketization on running time and user utility across Amazon, PriceRunner, Flights, and COMPAS.}
    \label{fig:bins-analysis}
\end{figure}

\subsection{Experiment Results}
\label{sec:experiment-results}
In this Section, we first present the results of our algorithm in Section~\ref{sec:possible-informative-equilibria}, whose goal is to detect trustworthy answers from data sources' results. Then present results on algorithms in Section~\ref{sec:maximally-influential} for finding influential queries.
\subsubsection{Detecting Trustworthy Answers. }
We present results on the scalability of Algorithm~\ref{alg:efficient-credibilility-filter}. The complexity of the algorithm depends on the number of tuples relevant to the user's intent ($z$) and the number of tuples returned by the data source ($k$). For instance, at the time of this writing, Amazon displays top-16 results per page; we therefore set $k=80$ to simulate a user viewing the first five pages of results. To maintain consistency across datasets of different sizes, we vary $z$ as a percentage of each dataset’s total size. As shown in Figure~\ref{fig:all_dataset_credible}, the runtime typically scales linearly with $z$. However, as discussed in Section~\ref{sec:possible-informative-equilibria}, we observe some sub-linear performance because the algorithm often terminates early.

\subsubsection{Influential Strategies. }  As explained in Section 5, the brute-force approach to find the maximally influential strategy is NP-hard due to the super-exponential complexity of the set of possible equilibria. Therefore, we only evaluate the efficient maximally influential algorithm ( MI; Algorithm~\ref{alg:max-info-merge-dp}) and the algorithm for detecting influential query ( DI; Algorithm~\ref{alg:detect-delta-relations}). 

\noindent
\textbf{Number of Attributes:}
The first set of experiments studies the effect of the size of the possible domain on the running time of the algorithms. To this end, we varied the number of attributes in the query from 1 to 3. The number of attributes (along with their cardinality) determines the search space. Figure~\ref{fig:time_vs_attrs} shows runtime as the number of \texttt{ORDER BY} attributes increases. Both algorithms slow down with larger datasets because the search space grows as the Cartesian product of attribute domains. The Amazon dataset incurs the highest cost owing to its large domain sizes, yet maximally influential (MI) still completes within the 15-minute cap for 3 attributes. The algorithm for finding the maximally influential strategy has a subroutine call Algorithm~\ref{alg:detect-delta-relations} and therefore has a higher complexity than detecting one influential query (DI).

\noindent
\textbf{Impact of Bucketization:}
To study the effect of this bucketization, we focus on the attributes with the largest pre-bucketization domains (Table~\ref{tab:datasets}): the continuous
attribute \texttt{price} for Amazon and Flights, the categorical attribute \texttt{product model} for
PriceRunner, and \texttt{age} for COMPAS. Census, which is a popular dataset used in prior works related to bias~\cite{jinyang2023biasranking,li2023queyrefinment}, did not require bucketization.

We consider two extremes for the data source's bias: (i) \emph{group-level} bias over coarse attribute
bands (e.g., the most expensive bin) and (ii) \emph{point-level} bias concentrated on specific values
(e.g., \$99.99). To interpolate between these extremes while controlling runtime, we use a
binary-tree (dyadic) refinement of the ordered domain. Specifically, we start with $B=2$ equal contiguous
bins and iteratively split each bin, yielding $B\in\{2,4,8,\dots\}$. For each bin level $B$, we run
Algorithm~\ref{alg:detect-delta-relations} and Algorithm~\ref{alg:max-info-merge-dp}, while evaluating the resulting query on the \emph{full} (unbucketized) domain. We enforce a 10-minute time budget per bin level.

Figure~\ref{fig:bins-analysis} summarizes the trade-offs across four datasets. For each dataset, we
report execution time and user utility. The time plots separate the cost of detecting an influential query from the end-to-end cost of computing the maximally
influential query. The utility plots measure the
effectiveness of the final query $q^\star$: we report the gain in user utility relative to the
original query (baseline) as a function of the additional relevant tuples recovered. As expected, the runtime grows
with the number of bins $B$, since it increases both the number of pairwise checks and the DP search space. At the same time, finer binning identifies more credible ranking information,
leading to queries that recover more relevant tuples. Empirically, the achieved utility is a monotone non-decreasing lower bound on the full-domain optimum.





\section{Related Work}
\label{sec:related-work}

\paragraph{ Responsible Data Management. }
As explained in Section 1, our work shares the same goal as efforts in responsible data management and science \cite{10.1145/3488717}. In particular, researchers have proposed several fair ranking schemes \cite{10.1145/3132847.3132938,10.1145/3366424.3380048,tabibian20a,mathioudakis2020affirmative,li2023detectiongroupsbiasedrepresentation}.
Instead of designing a fair ranking scheme, our goal is to investigate whether informative querying is possible in the presence of a conflict of interest and to enable users to extract reliable information from biased data sources.
Also, unlike our setting, these algorithms often require complete information about the underlying dataset.
Researchers have investigated the trade-off between the efforts of agents, e.g., school applicants, and their ranks in a ranked-based reward system, e.g., school admission, and proposed methods to design robust ranking systems ~\cite{liu2022strategicranking}.
However, we focus on the challenges of querying a data source where there are conflicts of interest between the user and the data source. Unlike our settings, the efforts of agents in \cite{liu2022strategicranking} may be highly costly, e.g., time on extracurricular activities.

\paragraph{Strategic Communication. } 
There is substantial research on strategic communication between agents with non-identical interests in economics and game theory \cite{10.2307/1913390,CHAKRABORTY200770,Blume2020EmpiricalCheapTalk,backus2019ebaycheaptalk,li2025ComparativeHousing}.
Our problem and approach have important differences with the settings of strategic communication in these fields.
First, we propose computational methods to find the influential equilibria and strategies for ranking queries.
Second, we address the challenges of querying with conflict of interest over large (structured) data.

\section{Appendix}
This section includes the proofs of our theoretical results.
\label{sec:proofs}
\subsection{Proofs for Section~\ref{sec:efficiently-detecting-uninformative-setting}}
\begin{proof}[\textbf{Proof of Theorem~\ref{theorem:noncred-babbling}}]The proof is simple, but we provide a longer proof to illustrate concepts.

Fix two set-equivalent intents $\tau \neq \tau'$ and two distinct interpretations
$\beta \neq \beta'$. Consider the induced $2\times 2$ game in which the
user can send two queries $q,q'$ and the data source chooses between
$\beta,\beta'$ after observing the query.

\medskip
\noindent
$(\Leftarrow)$ Suppose first that for both $t\in\{r,s\}$,
\[
U^t(\tau,\beta)\ge U^t(\tau,\beta')
\qquad\text{and}\qquad
U^t(\tau',\beta')\ge U^t(\tau',\beta).
\]
Define strategies
\[
P^r:\ \tau\mapsto q,\ \tau'\mapsto q',
\qquad
P^s:\ q\mapsto \beta,\ q'\mapsto \beta'.
\]
Since the two intents send different on-path queries (send with positive probability), Bayes' rule implies
\[
\Pr(\tau\mid q)=1,
\qquad
\Pr(\tau'\mid q')=1.
\]
Hence, after observing $q$, the data source's expected utility from choosing
$\beta$ is $U^s(\tau,\beta)$, while from choosing $\beta'$ it is
$U^s(\tau,\beta')$. By assumption,
\[
U^s(\tau,\beta)\ge U^s(\tau,\beta'),
\]
so $\beta$ is a best response after $q$. Similarly, after observing $q'$,
\[
U^s(\tau',\beta')\ge U^s(\tau',\beta),
\]
so $\beta'$ is a best response after $q'$.

Now consider the user's incentives. If intent $\tau$ follows $P^r$ and sends $q$,
it obtains utility $U^r(\tau,\beta)$. If it deviates and sends $q'$, the data source
responds with $\beta'$, yielding utility $U^r(\tau,\beta')$. Since
\[
U^r(\tau,\beta)\ge U^r(\tau,\beta'),
\]
intent $\tau$ has no profitable deviation (First bullet in Definition~\ref{definition:equilibrium}). Likewise, type $\tau'$ obtains
$U^r(\tau',\beta')$ from sending $q'$ and would obtain $U^r(\tau',\beta)$ by deviating
to $q$, and
\[
U^r(\tau',\beta')\ge U^r(\tau',\beta),
\]
so intent $\tau'$ also has no profitable deviation. Therefore $(P^r,P^s)$ is an
equilibrium, and it is influential since the two intents induce different
interpretations.

If instead for both $t\in\{r,s\}$,
\[
U^t(\tau,\beta')\ge U^t(\tau,\beta)
\qquad\text{and}\qquad
U^t(\tau',\beta)\ge U^t(\tau',\beta'),
\]
the same argument applies after swapping the labels of $\beta$ and $\beta'$.
Hence the interaction is influential.

\medskip
\noindent
$(\Rightarrow)$ Suppose the interaction is influential. Then there exists an
equilibrium in which two set-equivalent intents $\tau\neq\tau'$ induce two distinct
interpretations $\beta\neq\beta'$. Let the corresponding on-path queries be $q$ and
$q'$. Since the equilibrium is influential, the user's on-path behavior separates the
two intents. Without loss of generality, suppose $\tau$ sends $q$ and $\tau'$ sends
$q'$, and the data source responds with $\beta$ after $q$ and with $\beta'$ after
$q'$.

Because these queries are sent on path by distinct intents, Bayes' rule implies
\[
\Pr(\tau\mid q)=1,
\qquad
\Pr(\tau'\mid q')=1.
\]
Sequential rationality (second bullet in Definition~\ref{definition:equilibrium}) of the data source therefore requires
\[
U^s(\tau,\beta)\ge U^s(\tau,\beta')
\qquad\text{and}\qquad
U^s(\tau',\beta')\ge U^s(\tau',\beta),
\]
since $\beta$ must be a best response (optimal) after $q$ and $\beta'$ must be a best response
after $q'$.

Similarly, incentive compatibility for the user implies
\[
U^r(\tau,\beta)\ge U^r(\tau,\beta')
\qquad\text{and}\qquad
U^r(\tau',\beta')\ge U^r(\tau',\beta),
\]
because intent $\tau$ could deviate from $q$ to $q'$ and thereby induce $\beta'$, while
intent $\tau'$ could deviate from $q'$ to $q$ and thereby induce $\beta$.

Thus the first set of inequalities in the statement holds for both
$U^r$ and $U^s$. If the equilibrium instead matches $\tau$ with $\beta'$ and $\tau'$
with $\beta$, then the second set of inequalities holds. Therefore, the interaction is
influential if and only if one of the two stated conditions is satisfied.
\end{proof}

\begin{proof} [\textbf{Proof of Proposition~\ref{prop:symmetry}}]
It follows from Theorem 1 in \cite{CHAKRABORTY200770}.   
\end{proof}

\begin{proof} [\textbf{Proof of Theorem~\ref{prop:convex-saturation-tuples}}]
Fix any query $q$ and any tuple $e\in E$. Let $X_e:=r^\tau_e+b_e$.
Since $r^\tau_e\in[z]$, we have $X_e\in[1+b_e,\ z+b_e]$ almost surely under any belief.
For any interpretation rank $\tilde r^\beta_e\in[z]$, each conditional expected utility term is
\[
\mathbb E\!\left[L\!\Big(r^\tau_e-(\tilde r^\beta_e+b_e)\Big)\ \big|\ q\right]
=
\mathbb E\!\left[L\!\Big(\tilde r^\beta_e-X_e\Big)\ \big|\ q\right].
\]
By the hypothesis, for every $r\in[z]$ and every $\tilde r^\beta_e\in[z]$,
\[
L\!\Big(r-(r^{\beta^{\star}}_e+b_e)\Big)\le L\!\Big(r-(\tilde r^\beta_e+b_e)\Big).
\]
Substituting $r=r^\tau_e$ gives
\[
L(r^{\beta^{\star}}_e-X_e)\le L(\tilde r^\beta_e-X_e)
\]
Taking conditional expectations given $q$ yields
\[
\mathbb E\!\left[L(r^{\beta^{\star}}_e-X_e)\mid q\right]
\le
\mathbb E\!\left[L(\tilde r^\beta_e-X_e)\mid q\right]
\qquad\text{for all }\tilde r^\beta_e\in[z].
\]
Thus, for tuple $e$ the same interpretation rank $r^{\beta^{\star}}_e$ is optimal after every query $q$.
Since $U^s$ is additive across tuples, the data source has an optimal interpretation that assigns each $e$ the same
rank $r^{\beta^{\star}}_e$ independent of $q$. Therefore the data source response cannot depend on the query, and no equilibrium
can be influential.
\end{proof}

\begin{proof} [\textbf{Proof of Corollary~\ref{coro:quad-bias-saturation}}]
Fix any query $q$ and $e\in\tau(I)$. Since $r(e,\tau(I)) \in [k]$ , let $X_e:=\mathrm r(e,\tau(I))+b(e)\in[1+b(e),\,k+b(e)]$.
If $b(e)\ge k-\tfrac32$, then $X_e\ge k-\tfrac12$ , so $k$ is the closest feasible rank to $X_e$ and, since $L$ is
nondecreasing in $|\cdot|$, $r^\beta_e=k$ is optimal after every $q$.
If $b(e)\le \tfrac32-k$, then $X_e\le \tfrac32$, so $1$ is closest and $r^\beta_e=1$ is optimal after every $q$.
Thus each tuple has a query-independent optimal interpretation rank; by additivity, the entire interpretation is independent of $q$
.Hence no equilibrium can be influential.
\end{proof}

\subsection{Proofs for Section~\ref{sec:detecting-credibility}}
\begin{proof}[\textbf{Proof of Proposition~\ref{proposition:affine_boundary}}]
The proof is by algebraic simplification.

For (i), write $r_e:=\mathrm r(e,\tau)$, $\beta_e:=\mathrm r(e,\beta)$, $\beta'_e:=\mathrm r(e,\beta')$. Then
\[
u^r(r_e,\beta_e)-u^r(r_e,\beta'_e)
=-(r_e-\beta_e)^2+(r_e-\beta'_e)^2
=2r_e(\beta_e-\beta'_e)-(\beta_e^2-\beta_e'^2).
\]
Summing over $e$ gives
\[
U^r(\tau,\beta)-U^r(\tau,\beta')=\sum_{e\in E}\gamma_e r_e-\sum_{e\in E}(\beta_e^2-\beta_e'^2),
\]
so indifference is equivalent to $\sum_{e\in E}\gamma_e\,\mathrm r(e,\tau)=\alpha^r$.

For (ii),
\[
u^s_e(r_e,\beta_e)-u^s_e(r_e,\beta'_e)
=\gamma_e r_e-\Big((\beta_e+b_e)^2-(\beta'_e+b_e)^2\Big).
\]
Summing and taking expectation under $\Phi$ yields
\[
\mathbb E_{\tau\sim\Phi}[U^s(\tau,\beta)-U^s(\tau,\beta')]
=\sum_{e\in E}\gamma_e\,\Phi_e-\alpha^s,
\]
where $\alpha^s=\sum_e\big((\beta_e+b_e)^2-(\beta'_e+b_e)^2\big)=\alpha^r+2\sum_e b_e(\beta_e-\beta'_e)$.
Thus the data source is indifferent iff $\sum_{e\in E}\gamma_e\,\Phi_e=\alpha^s$.
\end{proof}

\begin{proof}[\textbf{Proof of Lemma~\ref{lemma:two_tuple_shift}}]
Under the corollary’s assumptions, $\gamma_{\bar e}=0$ for $\bar e\notin\{e,e'\}$ and
$\gamma_{e'}=-\gamma_e\neq 0$, so the user indifference condition in Proposition~\ref{proposition:affine_boundary} reduces to
$\mathrm r(e,\tau)-\mathrm r(e',\tau)=\frac{\alpha^{r}}{\gamma_e}$.
Moreover, only $e,e'$ contribute to the intercept, and using
$\mathrm r(e',\beta)-\mathrm r(e',\beta')=-(\mathrm r(e,\beta)-\mathrm r(e,\beta'))$ gives
\[
\alpha^{s}
=\alpha^{r}+2\big(b(e)-b(e')\big)\big(\mathrm r(e,\beta)-\mathrm r(e,\beta')\big)
=\alpha^{r}+\gamma_e\big(b(e)-b(e')\big).
\]
Thus $\alpha^{r}=\alpha^{s}-\gamma_e(b(e)-b(e'))$, and substituting yields
\[
\mathrm r(e,\tau)-\mathrm r(e',\tau)
=\frac{\alpha^{s}}{\gamma_e}-(b(e)-b(e')).
\]
This establishes the claimed threshold form \eqref{equation:indifference-threshold}.
\end{proof}

\begin{proof} [\textbf{Proof of Proposition~\ref{prop:intents_change_interpretation}}]
For an intent $\tau$, define $d(\tau)=\mathrm r(e,\tau(I))-\mathrm r(e',\tau(I))$.
By Lemma~\ref{lemma:two_tuple_shift}, the user's indifference between $\beta$ and $\beta'$ is determined solely by whether $d(\tau)$ is above or below the threshold. The unbiased case uses threshold $\delta_0$ and the biased case uses threshold $\delta$.
If $d(\tau)\le \min\{\delta_0,\delta\}$, then $d(\tau)\le \delta_0$ and $d(\tau)\le \delta$, so $\tau$ is on the same (weak) side of both boundaries and induces the same interpretation in both bias and unbiased cases. If $d(\tau)>\max\{\delta_0,\delta\}$, then $d(\tau)>\delta_0$ and $d(\tau)>\delta$, so $\tau$ is again on the same side of both boundaries and induces the same interpretation in both cases. Therefore, the induced interpretation can change only when the relative position of the tuples lies between the two thresholds.
\end{proof}

\begin{proof}[\textbf{Proof of Theorem~\ref{theorem:closed-form-trustworthy-information}}]
Let $r=\mathrm r(e,\tau(I))$ and $r'=\mathrm r(e',\tau)$, where smaller rank is better and a fixed $\delta\in\{1,\dots,z-1\}$. Each $\delta$ partitions the space of all possible ranks of two tuples $(r,r')\in[z]^2$ into two sets $\mathcal H$ and $\mathcal H'$. Intuitively, $\mathcal H$ corresponds to rank pairs where the difference $r' - r$ is at least $\delta$, and $\mathcal H'$ is its complement. 
We consider the $\delta$-threshold partition that favors $e$:
\[
\mathcal{H}(\delta) = \{(r,r') \in [z]^2 : r'-r \ge \delta\}, \qquad 
\mathcal{H}'(\delta) = [z]^2 \setminus \mathcal{H}(\delta).
\]
Thus, on $\mathcal H(\delta)$, tuple $e$ outranks $e'$ by at least $\delta$ intent ranks.

Conditioning the (uniform) prior on $\mathcal H$ and $\mathcal H'$ yields posterior mean ranks. Let $\Phi_e(\mathcal{H})=\mathbb E[r\mid (r,r')\in \mathcal{H}]$ and
$\Phi_{e'}(\mathcal{H})=\mathbb E[r'\mid (r,r')\in \mathcal{H}]$.
A direct count over the discrete triangle $\{r'-r\ge\delta\}$ gives
\[
\Phi_e(\mathcal{H})=\frac{z-\delta+2}{3},\qquad 
\Phi_{e'}(\mathcal{H})=\frac{2z+\delta+1}{3}.
\]

For the complement means $\Phi_e(\mathcal{H}')=\mathbb E[r\mid (r,r')\in \mathcal{H}']$ and
$\Phi_{e'}(\mathcal{H}')=\mathbb E[r'\mid (r,r')\in \mathcal{H}']$, using total sums on $[z]^2$ and simplifying yields
\[
\Phi_e(\mathcal{H}')=\frac{\delta^3 - 3(z+1)\delta^2 + (3z^2+6z+2)\delta + 2z(z^2-1)}{d(z,\delta)},
\]
\[
\Phi_{e'}(\mathcal{H}')=\frac{-\delta^3 + 3\delta z^2 + (3z^2+1)\delta + z(z^2-1)}{d(z,\delta)}.
\]
In particular, the difference across the complement region simplifies exactly to:
\begin{equation}
\label{complement-difference}
\Phi_e(\mathcal{H}')-\Phi_{e'}(\mathcal{H}')=s(z,\delta).
\end{equation}

Assume an unconstrained, real-valued data source interpretation $\mathrm r(\cdot,\beta)\in\mathbb R$ (no rounding/clipping) for simplicity, with quadratic utility
$u^s_x(r^\tau_x,r^\beta_x)=-(r^\tau_x-(r^\beta_x+b_x))^2$.
The unique optimal posterior interpretation after observing $X \in \{\mathcal{H}, \mathcal{H}'\}$ is:
\begin{equation}
\mathrm r(x,\beta_X)=\Phi_x(X)-b_x,\qquad x\in\{e,e'\}.
\label{eq:best_response}
\end{equation}
Although these scores are real-valued, the interpretation result $\beta(I)$ induces a discrete ordering by sorting tuples by $\mathrm r(\cdot,\beta_X)$ (and returning the top-$k$). Hence, on $\mathcal H'$, tuple $e$ is ranked ahead of $e'$ iff
$\mathrm r(e,\beta_{\mathcal{H}'}) < \mathrm r(e',\beta_{\mathcal{H}'})$.

Under $u^r(r^\tau_x,r^\beta_x)=-(r^\tau_x-r^\beta_x)^2$, user indifference between the two posterior interpretations
$\beta_{\mathcal{H}}$ and $\beta_{\mathcal{H}'}$ yields an affine boundary in $(r,r')$, which can be written as a threshold on $(r'-r)$.
In a threshold equilibrium over the discrete integer grid $[z]^2$, any real indifference threshold $t \in \mathbb{R}$ induces exactly the same partition $\mathcal{H}(\delta)$ whenever $t\in(\delta-1,\delta]$.
Evaluating the user's indifference condition yields
\[
t=\frac{(\mathrm r(e',\beta_{\mathcal{H}})+\mathrm r(e',\beta_{\mathcal{H}'}))-(\mathrm r(e,\beta_{\mathcal{H}})+\mathrm r(e,\beta_{\mathcal{H}'}))}{2}
\]
\[
=\delta-g(z,\delta)+(b(e)-b({e'})),
\]
where substituting the posterior means and the optimal data source's posterior interpretation (Equation~\ref{eq:best_response}) gives $g(z,\delta)$ as stated.
Hence, the required separation condition $t\in(\delta-1,\delta]$ holds if and only if
\begin{equation}
b(e)-b({e'}) \in \big(g(z,\delta)-1,\ g(z,\delta)\big].
\label{eq:consistency}
\end{equation}

If $\delta > 1$, the complement region $\mathcal{H}'(\delta)$ contains the nonempty strip
$
\mathcal S(\delta)=\{(r,r')\in[z]^2:\ 1\le r'-r\le \delta-1\}\subseteq \mathcal{H}'(\delta),$
on which $r<r'$ (i.e., $e$ truly outranks $e'$) and also includes all intents with $r>r'$ (where $e'$ truly outranks $e$).
Tuple $e$ is non-credible witnessed by $e'$ if, under the posterior interpretation induced by $\mathcal{H}'(\delta)$, the data source ranks $e$ ahead of $e'$:
\[
\mathrm r(e,\beta_{\mathcal{H}'})<\mathrm r(e',\beta_{\mathcal{H}'})
\iff
b(e)-b({e'})>\Phi_e(\mathcal{H}')-\Phi_{e'}(\mathcal{H}')
\]
Substituting the difference across the complement region from Equation~\eqref{complement-difference} into this inequality yields: $b(e)-b({e'})>s(z,\delta).$
Intersecting this with Equation~\ref{eq:consistency} gives the bias interval:$$
b(e)-b({e'}) \in \Big(\max\{g(z,\delta)-1,\ s(z,\delta)\},\ g(z,\delta)\Big]
$$

\end{proof}

\subsection{Proofs for Section 5.2}

\begin{proof} [\textbf{Proof of Theorem~\ref{theorem:equilibrium}
}]
The proof is by construction, extending the results of Proposition~\ref{prop:one_delta_influential}.
Fix an ordered pair $(e,e')$ with $\delta:=\delta^\star(e,e')\in\{1,\dots,z-1\}$ and consider the two
regions of the space of possible rank differences.
\[
\mathcal H \;=\;\{(r,r')\in[z]^2:\ r'-r\ge \delta\},
\qquad
\mathcal H' \;=\; [z]^2\setminus \mathcal H
\ \ .
\]
By Proposition~\ref{prop:one_delta_influential}, there exists an influential equilibrium supported on
the two queries that separate $\mathcal H$ from $\mathcal H'$. Under our query language, the user must submit a single ranked list. Since both $\mathcal H$ and
$\mathcal H'$ are nonempty, we can select two executable queries $q_{\mathcal H}$ and $q_{\mathcal H'}$
whose induced rankings are consistent with $\mathcal H$ and $\mathcal H'$, respectively. We define the
data source's on-path responses to $q_{\mathcal H}$ and $q_{\mathcal H'}$ to match the equilibrium
responses guaranteed by Proposition~\ref{prop:one_delta_influential}. This preserves the best responses
and incentive constraints (two bullets in Definition~\ref{definition:equilibrium}), so it yields an influential equilibrium under executable queries. Finally, allowing the user to submit any super-rank of $q^{\mathrm{base}}$ does not eliminate this
equilibrium. We construct equilibrium strategies, the user startegy can place all probability on the two on-path queries $q_{\mathcal H}$ and
$q_{\mathcal H'}$, while any other super-rank queries are assigned probability $0$ and the data
source's interpretation to them is specified off-path. Hence, an influential equilibrium exists under a
super-rank strategy.
\end{proof}

\begin{proof} [\textbf{Proof of Proposition~\ref{prop:enumerate-complexity}}]
Let $m$ be the number of tuples in $q^{\mathrm{base}}$ and let $r=|R(I)|$ be the number of tuples that may be appended. The algorithm outputs exactly $2^{\,m-1}\,a(r)$ super-ranks, where $a(r)=\sum_{k=1}^r k!\,S(r,k)$ is the ordered Bell (Fubini) number. Consequently the running time is $\Theta\!\big(2^{\,m-1}\,a(r)\,(m+r)\big)$ and is \emph{super-exponential} in $r$ (since $a(r)\ge r!$).
We analyze the two phases of the algorithm separately and then combine them.

(\emph{Merge-only count.}) The tuples of $q^{\mathrm{base}}$ form a total preorder with $m-1$ adjacent boundaries. Any merge-only coarsening corresponds to choosing which boundaries to delete; there are $2^{m-1}$ choices, each yielding a unique coarsening.

(\emph{Append count.}) Appending $r$ labeled items as a ranked/tied suffix is in bijection with weak orders (total preorders) on an $r$-set. The number of such weak orders is
$a(r)=\sum_{k=1}^r k!\,S(r,k)$, where $S(r,k)$ counts partitions of the $r$ items into $k$ unlabeled blocks and the factor $k!$ orders these blocks into ranks.~\cite{stirlingnumbers}

(\emph{Product.}) The suffix is independent of the chosen merge-only coarsening, so the total number of super-ranks is the product $2^{m-1}\,a(r)$.

(\emph{Running time.}) The algorithm generates each coarsening once (BFS over adjacent merges), and for each coarsening enumerates all $a(r)$ weak orders of the $r$ appendable items. Materializing one output takes $O(m+r)$ time (copying/coalescing $m$ tuples and writing the $r$-item suffix). Hence the time is $\Theta(2^{m-1}\,a(r)\,(m+r))$.

(\emph{Super-exponential growth in $r$.}) Since $S(r,r)=1$, the $k=r$ term gives $a(r)\ge r!$, while $a(r)\le \sum_{k=1}^r k^r \le r^{r+1}$. Thus $a(r)$ (and therefore the total) grows faster than any $c^r$ for fixed $c>1$, i.e., super-exponentially in $r$.
\end{proof}

\begin{proof}[\textbf{Proof of Theorem~\ref{thm:superrank-nphard}}]
We show NP-hardness by a Karp reduction from \textsc{Subset-Sum}.

\paragraph{\textsc{Subset-Sum}.}
Given positive integers $w_1,\dots,w_n$ and a target $W$, decide whether there exists
$S\subseteq\{1,\dots,n\}$ such that $\sum_{j\in S} w_j = W$.

\paragraph{Decision variant (\textsc{MaxInf}).}
Given a base ranking $q^{\min}=e_1\prec e_2\prec\cdots\prec e_m$, utilities $U^s,U^r$ (specified below),
and a threshold $T\in\mathbb Z$, decide whether there exists a \emph{merge-only} super-rank query
$q\in \mathcal Q_{\mathrm{merge}}^{*}(q^{\min})$ such that
$U^r\!\bigl(q,\beta^{*}(q)\bigr)\ge T$.
Here $\mathcal Q_{\mathrm{merge}}^{*}(q^{\min})$ denotes the set of super-rank queries obtained from $q^{\min}$
by merging any subset of adjacent boundaries (i.e., partitioning $e_1,\dots,e_m$ into contiguous tie-blocks),
and $\beta^*(q)=\arg\max_{\beta} U^s(q,\beta)$.

Since $\mathcal Q_{\mathrm{merge}}^{*}(q^{\min})\subseteq \mathcal Q_{\mathrm{super}}(q^{\min})$,
NP-hardness for \textsc{MaxInf} over $\mathcal Q_{\mathrm{merge}}^{*}$ implies NP-hardness for the full
super-rank problem.

\paragraph{Reduction.}
Given an instance $(w_1,\dots,w_n,W)$, set $m:=n+1$ and construct
\[
q^{\min}=e_1\prec e_2\prec\cdots\prec e_m.
\]
Any $q\in\mathcal Q_{\mathrm{merge}}^{*}(q^{\min})$ is uniquely encoded by a binary vector
$b=(b_1,\dots,b_n)\in\{0,1\}^n$, where $b_j=0$ means the boundary between ranks $j$ and $j{+}1$ is merged (tied),
and $b_j=1$ means it is kept as a cut.

Fix an arbitrary interpretation $\beta^\dagger$ (e.g., a canonical interpretation).
Define the data source utility to make $\beta^\dagger$ the unique best response:
\[
U^s(q,\beta):=
\begin{cases}
1 & \text{if }\beta=\beta^\dagger,\\
0 & \text{otherwise}.
\end{cases}
\]
Hence $\beta^*(q)=\beta^\dagger$ for all $q$.

Define the user utility (computable in polynomial time from $b$) by
\[
U^r(q,\beta^\dagger)\;:=\;-\left|\sum_{j:\,b_j=0} w_j \;-\; W\right|.
\]
Finally set the threshold $T:=0$.

\paragraph{Correctness.}
For any merge-only super-rank query $q$ encoded by $b$,
\[
U^r\!\bigl(q,\beta^*(q)\bigr)\ge 0
\iff
-\left|\sum_{j:\,b_j=0} w_j - W\right|\ge 0
\iff
\sum_{j:\,b_j=0} w_j = W.
\]
Thus, letting $S=\{\,j : b_j=0\,\}$, there exists $q$ with $U^r(q,\beta^*(q))\ge 0$
if and only if the \textsc{Subset-Sum} instance has a solution.

\paragraph{Complexity.}
The construction of $q^{\min}$ and the explicit utilities takes $\mathrm{poly}(n)$ time.
Moreover, given $b$, evaluating $U^r$ requires computing $\sum_{j:b_j=0}w_j$, which is $\mathrm{poly}(n)$.
Therefore \textsc{Subset-Sum} reduces to \textsc{MaxInf} in polynomial time, proving NP-hardness.
\end{proof}

\begin{proof}[\textbf{Proof of Lemma~\ref{lem:run-scores}}]
Fix $q'\in\mathcal Q_{\mathrm{merge}}(q)$ and write $\Pi(q')=\{[i_1,j_1],\dots,[i_t,j_t]\}$ for its
merged-interval partition. By additivity of $U^r$, the expected utility decomposes as a sum over
the sets of tuples whose ranks lie in each interval $[i_\ell,j_\ell]$. Under the i.i.d.\ prior,
tuples within the same merged interval are conditionally exchangeable, so each contributes the same
expected per-tuple utility $\overline u_{q'}(i_\ell,j_\ell)$, yielding a contribution
$n_{i_\ell:j_\ell}\overline u_{q'}(i_\ell,j_\ell)$ from that interval.
Finally, merging intervals disjoint from $[i,j]$ only changes ties \emph{within} those disjoint
intervals and imposes no additional ordinal constraints involving tuples whose original ranks lie in
$[i,j]$. Hence the posterior (and the data source’s tuplewise posterior interpretaion) for tuples in $[i,j]$ is
the same under $q'$ as under the single-interval merge query $q_{i:j}$. Therefore
$\overline u_{q'}(i,j)=u_{ij}$ for every $[i,j]\in\Pi(q')$, and summing over the intervals
gives the stated form.
\end{proof}

\begin{proof} [\textbf{Proof of Proposition~\ref{thm:dp-recurrence}}]
Any optimal solution on the prefix $[1..j]$ ends with some last run $[i..j]$. By Lemma~\ref{lem:run-scores}, its contribution is $Score(i,j)$, and the best possible contribution before it is $\mathrm{OPT}(i-1)$. Maximizing over the choice of $i$ yields the recurrence. Conversely, picking the maximum $i$ for each $j$ constructs an optimal partition.
\end{proof}

\bibliographystyle{ACM-Reference-Format}
\bibliography{sample}

\end{document}